 \documentclass[draft]{article}

\usepackage{amsmath,amsfonts,amsthm,amssymb,amscd,color}
\renewcommand{\Re}{\mathop{\rm Re}\nolimits}
\renewcommand{\Im}{\mathop{\rm Im}\nolimits}

\theoremstyle{plain} \newtheorem{theorem}{Theorem}[section]
\newtheorem{lemma}[theorem]{Lemma}
\newtheorem{proposition}[theorem]{Proposition}
\newtheorem{corollary}[theorem]{Corollary} \theoremstyle{definition}
\newtheorem{definition}[theorem]{Definition} \theoremstyle{remark}
\newtheorem{assumption}[theorem]{Assumption} \theoremstyle{remark}
\newtheorem{remark}[theorem]{Remark} 
 
\newcommand{\R}{{\mathbb R}}

\newcommand{\Z}{{\mathbb Z}}

\newcommand{\N}{{\mathbb N}}

 \def\im{{\rm i}}

\newcommand{\C}{\mathbb{C}} \newcommand{\T}{\mathbb{T}}

\newcommand{\cC}{\mathcal{C}}

\newcommand{\Cp}{\C_p} 
\newcommand{\Tp}{\T_+}

\def\({\left(}
\def\){\right)}
\def\<{\left\langle}
\def\>{\right\rangle}

\numberwithin{equation}{section}

\begin{document}

\title{Dispersive estimates for quantum walks on 1D lattice}

\author {Masaya Maeda, Hironobu Sasaki, Etsuo Segawa, Akito Suzuki, Kanako Suzuki }

\maketitle

\begin{abstract}
We consider quantum walks with position dependent coin on 1D lattice $\Z$.
The dispersive estimate $\|U^tP_c u_0\|_{l^\infty}\lesssim (1+|t|)^{-1/3} \|u_0\|_{l^1}$ is shown under $l^{1,1}$ perturbation for the generic case and $l^{1,2}$ perturbation for the exceptional case, where $U$ is the evolution operator of a quantum walk and $P_c$ is the projection to the continuous spectrum.
This is an analogous result for Schr\"odinger operators and discrete Schr\"odinger operators.  
The proof is based on the estimate of oscillatory integrals expressed by Jost solutions.
\end{abstract}

\section{Introduction and main results}

In this paper, we study space-time discrete unitary dynamics called quantum walks (QWs).
QWs  are quantum analog of classical random walks \cite{ABNVW01,FH10Book,Gudder88Book,Meyer96JSP}.
QWs  are now attracting diverse interest due to its connections to various regimes of mathematics and physics such as quantum search algorithms \cite{AKR05Proc,Childs09PRL,Portugal13Book} and topological insulators \cite{AO13PRB,CGSVWW16JPA,EKOS17JPA, FFS18JMP,
	GNVW12CMP,Kitaev06AP,Kitagawa12QIP,KRBD10PRA}. 
Moreover, QWs  have been realized experimentally by optical lattice \cite{Karski09Science}, photons \cite{SCPGJS11,Schreiberetal12Science} and ion trappping \cite{ZKGSBR10PRL} (see \cite{MW14Book} for more reference).

We first give a formal definition of QWs on $\Z$.
For a $\C^2$-valued map $u:\Z\to \C^2$, we write
\begin{align*}
u(x)=\begin{pmatrix}
u_\uparrow(x)\\ u_\downarrow(x)
\end{pmatrix} \ \text{ and } \|u(x)\|_{\C^2}^2:=|u_\uparrow(x)|^2+|u_\downarrow(x)|^2.
\end{align*}
For a Banach space $X$ (usually $\C$ or $\C^2$), we set
\begin{align*}
l^{p,\sigma}(\Z,X):=\left\{ u:\Z\to X\ |\ \|u\|_{l^{p,\sigma}}<\infty \right\}\ \text{where } \|u\|_{l^{p,\sigma}}^p:=\sum_{x\in \Z}\<x\>^{p\sigma}\|u(x)\|_X^p,
\end{align*}
$l^p(\Z,X):=l^{p,0}(\Z,X)$ and $\mathcal{H}:=l^2(\Z,\C^2)$ where $\<x\>^2:=1+|x|^2$.
When $X$ is clear from the context, we write $l^{p,\sigma}$ instead of $l^{p,\sigma}(\Z,X)$.
For Banach spaces $X,Y$, we set $\mathcal L(X,Y)$ to be the Banach space of all bounded linear operators from $X$ to $Y$.
Further, we set $\mathcal L(X):=\mathcal L(X,X)$.

The shift operator $S\in \mathcal L(\mathcal{H})$ and the coin operator $\hat C \in \mathcal L(\mathcal{H})$ are given by
\begin{align*}
S=\begin{pmatrix}\tau & 0 \\ 0 & \tau^{-1} \end{pmatrix},\quad \(\hat C u\)(x)=C(x)u(x),
\end{align*}
where $\(\tau u\)(x)=u(x- 1)$ and $C:\Z \to U(2):=\{C:\ 2\times 2\text{ unitary operator}\}$.
In the following, by abuse of notation, we sometimes write $\hat C=C(x)$.

Given $S$ and $\hat C$, QWs  on a lattice $\Z$ are unitary evolution dynamics on $\mathcal{H}$ generated by unitary operators with the form
\begin{align*}
U=S\hat C.
\end{align*}
That is, the quantum walker $u(t)$ with initial state $u(0)=u_0\in \mathcal{H}$ is given by
\begin{align*}
u(t)=U^t u_0,\ t\in\Z.
\end{align*}

We are interested in the long time behavior of QWs with coins $C(x)$ that converge to some constant coin $C_0$ as $|x|\to \infty$.
There are many ways understanding the dynamics of QWs such as seeking a weak limit theorem 
\cite{Endo16IIS, FFS18, Konno02QIP, MSSSS18QIP,
SST18LMP}, 
considering the asymptotic velocity \cite{GJS04, Suzuki16QIP}, and studying the spectral stability (or scattering theory) 
\cite{ABJ15RMP, SS16QS, SST17LMP}.
In this paper, we will consider the dispersive estimate of QWs, which is the $l^\infty$--$l^1$ decay estimate:
\begin{align}\label{Dispersive}
\|U^tP_c u_0 \|_{l^\infty}\lesssim \<t\>^{-1/3} \|u_0\|_{l^1},\quad t\in \Z,
\end{align}
where $P_c$ is the projection to the continuous spectrum (see \eqref{def:proj} below) and $a\lesssim b$ means $a \leq C b$ with some positive constant $C$.

Dispersive estimate \eqref{Dispersive} provides quantitative information of the continuous component $P_c u(t)$ of the state.
Note that we cannot obtain dispersive estimate only from the absolute continuity of the continuous spectrum, which is a qualitative information.
In this sense, dispersive estimate gives us deeper understanding of the dynamics of QWs.
However, this is not the only significance of dispersive estimate.
In the field of nonlinear dispersive partial differential equations, dispersive estimates are important because they provide the Strichartz estimates \cite{KT98AJM} which are the fundamental inequalities for the study of the well-posedness and asymptotic behavior of the solutions.
Moreover, dispersive estimate plays a crucial role for the study of  asymptotic stability of solitons \cite{JSS91CPAM,SW92JDE}.
For discrete dispersive equations (the space is discrete but the time is continuous) such as discrete Schr\"odinger equations (also called continuous time QWs), dispersive estimate was studied by \cite{CT09SIAM, EKT15JST, PS08JMP, SK05N}, with direct application to the stability of solitons \cite{Bambusi13CMPdb, CT09SIAM, KPS09SIAM, MP12DCDS}.
Recently, nonlinear QWs, which are QWs with the coin operator depending on the state (see \cite{MSSSS18DCDS}), have been proposed by several authors \cite{GTB16PRA, LKN15PRA, NPR07PRA}.
Since soliton phenomena are observed numerically \cite{LKN15PRA,MSSSS19JPC, NPR07PRA}, it seems to be important to prepare dispersive estimate for QWs, which should play a crucial role for the study of asymptotic stability of solitons of nonlinear QWs.

For QWs with constant coin $C(x)\equiv C_0$, estimate \eqref{Dispersive} was shown by Sunada-Tate \cite{ST12JFA} (see also \cite{MSSSS18DCDS}).
In \cite{MSSSS18DCDS}, it was applied to nonlinear QWs with weak nonlinearity to show the scattering of the solutions.

We prepare several notations and assumptions to state our result precisely.
Recall that for $C\in U(2)$, there exists unique (up to the symmetry $(\alpha,\beta,\theta)\mapsto (-\alpha,-\beta,\theta+\pi)$) $(\alpha,\beta,\theta)\in \C^2\times \T$ ($\T:=\R/2\pi\Z$) satisfying $|\alpha|^2+|\beta|^2=1$ s.t.
\begin{align*}
C=C_{\alpha,\beta,\theta}:=e^{\im \theta}\begin{pmatrix}
\beta & \overline{\alpha}\\ -\alpha & \overline{\beta}
\end{pmatrix}.
\end{align*}

For the limiting coin $C_0$ and the coin operator $\hat C$, we will always assume the following:
\begin{assumption}[$l^1$ perturbation]\label{ass:1}
	We assume $C_0=C_{\alpha_0,\beta_0,0}$ and $C(x)=C_{\alpha(x),\beta(x),\theta(x)}$ with $0<|\alpha_0|<1$ and $|\alpha(x)|<1$ for all $x\in \Z$ and  $\| \alpha(\cdot)-\alpha_0\|_{l^1}+\|\theta(\cdot)\|_{l^1}<\infty$.
\end{assumption}

\begin{remark}
	There is no loss of generality restricting $C_0=C_{\alpha_0,\beta_0,0}$ instead of $C_0=C_{\alpha_0,\beta_0,\theta_0}$.
\end{remark}

\begin{remark}
	The condition $|\alpha(x)|\neq 1$ for all $x\in \Z$ is not technical.
	If there exists $x_0\in \Z$ s.t.\ $|\alpha(x_0)|=0$ (and $\beta(x_0)=0$), the QW will be decoupled to two noninteracting systems.
	See, section \ref{sec:decoup} in the appendix of this paper. 
\end{remark}



We first prove the basic spectral property for $l^1$ perturbation.
It is well known that 
$$\sigma_{\mathrm{ess}}(U)=\{e^{\im \lambda}\ |\ \lambda\in\R,\  |\cos \lambda|\leq \sqrt{1-|\alpha_0|^2}\},$$
where $\sigma_{\mathrm{ess}}(U)$ is the essential spectrum of $U$.

\begin{definition}
	We will call $\{e^{\im \lambda}\ |\ \lambda\in\R,\  |\cos \lambda| < \sqrt{1-|\alpha_0|^2}\}$, the interior of the essential spectrum and $\{e^{\im \lambda}\ |\ \lambda\in\R,\  |\cos \lambda| = \sqrt{1-|\alpha_0|^2}\}$ the edge of the essential spectrum.
	Note that the edge of essential spectrum consist of 4 points.
\end{definition}

\begin{proposition}\label{prop:spec}
Under Assumption \ref{ass:1}, we have the following:
\begin{itemize}
\item[(i)]
$U$ has no singular continuous spectrum.

\item[(ii)]
There exists no eigenvalue in the interior of the essential spectrum.
Moreover, under the stronger assumption $\|\alpha(\cdot)-\alpha_0\|_{l^{1,1}}+\|\theta(\cdot)\|_{l^{1,1}}<\infty$, there exists no eigenvalue on the edge of the essential spectrum.
\item[(iii)]
If $w$ is an eigenvalue of $U$, it is simple. That is, $\mathrm{dim}\ \mathrm{ker} (U-w) =1$.
\end{itemize}
\end{proposition}

\begin{remark}
The absence of singular continuous spectral has been proven under a slightly stronger condition \cite{ABJ15RMP}.
The absence of embedded eigenvalue seems to be a new result.
\end{remark}

We introduce the notation of resonance.
\begin{definition}\label{def:resonance}
Let $w$ be on the edge of essential spectrum of $U$.
If there exists a bounded solution of $(U-w)u=0$, we say $w$ is a resonance.
If all 4 points of the edge of the essential spectrum of $U$ are not resonance, we say $U$ is generic.
If not, we say $U$ is exceptional.  
\end{definition}

An alternative definition of resonance will be given in Lemma \ref{lem:17}.

Under stronger assumption, we have further information of the spectrum.
\begin{proposition}\label{prop:finite}
Assume $\|\alpha(\cdot)-\alpha_0\|_{l^{1,k}}+\|\theta(\cdot)\|_{l^{1,k}}<\infty$ with $k=1$ for the generic case and $k=2$ for the exceptional case.
Then, $\sigma_{\mathrm{d}}(U)$, the set of discrete spectrum of $U$, is a finite set.
\end{proposition}

If the discrete spectrum is finite set, then we have $$\mathrm{dist}(\sigma_{\mathrm{ess}}(U),\sigma_{\mathrm{d}}(U)):=\inf\{|w_1-w_2|\ |\ w_1\in \sigma_{\mathrm{ess}}(U),\ w_2\in \sigma_{\mathrm{d}}(U)\}>0.$$
By using this gap, we can define the Riesz projection to the continuous spectrum (which is the essential spectrum because of the absence of embedded eigenvalue Proposition \ref{prop:spec} (ii)) by
\begin{align}\label{def:proj}
P_c:=P_+ + P_- :=\sum_{\pm}\frac{\im}{2\pi }\int_{\Gamma_\pm}(U-w)^{-1}\,dw,
\end{align}
where $\Gamma_\pm=\{\gamma_\pm(t)\ |\ t\in\T\}$ is an smooth closed curve encircling $\sigma_\pm:=\sigma_{\mathrm{ess}}(U)\cap \{\pm \Im z>0\}$ in anti-clockwise direction with $\mathrm{dist}(\sigma_\pm,\gamma_\pm)<\mathrm{dist}(\sigma_{\mathrm{ess}}(U),\sigma_{\mathrm{d}}(U))$ for all $t\in\T$.
It is known that $P_c$ is the projection to the spectral component which is encircle by the curves (for this case $\sigma_{\mathrm{ess}}$), for details see Part I, Chapter I of \cite{GohbergBookI}.

By preparing the projection $P_c$, we are now in the position to state our main result.
\begin{theorem}\label{thm:main}
Under the assumption of Proposition \ref{prop:finite}, we have the dispersive estimate \eqref{Dispersive}.
\end{theorem}

We now explain the strategy of the proof of Theorem \ref{thm:main}.
Since the problem is $1$ dimensional, it is natural to try to express the integral kernel of the generator $U^tP_c$ by Jost solutions.
Indeed, for Schr\"odinger equations \cite{Weder00JFA, GS04CMP} and discrete Schr\"odinger equations \cite{PS08JMP, CT09SIAM, EKT15JST} dispersive estimate are obtained by such method.

To obtain the Jost solutions, which is the solutions of $(U-e^{\im \lambda})u=0$ having specific behavior in the spatial infinity, we employ the ``CMV" representation of QWs $\cite{CGMV10CPAM}$, which is actually a slight change of view from the vertices to the edges of the graph $\Z$.
By this ``shift of view" one can obtain transfer matrix representation of the problem $(U-\lambda)u=0$, which can be viewed as an evolution equation with respect to the $x$ variable.
The construction of Jost solutions for $\lambda$ not at the edge of essential spectrum is standard, one simply diagonalizes the transfer matrix and solves the $l^1$ perturbation problem.
However, for the dispersive estimate, we need uniform estimates (especially near the edge of the essential spectrum) of the Jost solutions.
To obtain estimates of the transfer matrices near the edge that cannot be diagonalized, we triangulate it instead.
By this simple but new idea, we can obtain the uniform estimate of Jost solutions.

This paper is organized as follows.
In section \ref{sec:gauge}, we simplify the coin operator by introducing a Gauge transformation.
In section \ref{sec:2}, we reformulate the generalized eigenvalue problem $(U-\lambda)u$ by using transfer matrix.
We also give an expression of the kernel of the resolvent in terms of solutions of the generalized eigenvalue problems.
In section \ref{sec:3}, we study the case of constant coin.
In section \ref{sec:4}, we show the existence and give  estimates for Jost solutions.
This section is the main part of this paper.
In section \ref{sec:scat}, we develop a stationary scattering theory and prove Proposition \ref{prop:spec} (i) and (ii) and Proposition \ref{prop:finite}.
In section \ref{sec:6} we prove Theorem \ref{thm:main}.

\section{Gauge transformation}\label{sec:gauge}
Before going into the spectral analysis, we show that we can slightly simplify our coin operator.

First, if $G\in \mathcal{L}(\mathcal{H})$ has the following form, we call $G$ a gauge transformation:
\begin{align*}
\(Gu\)(x)=\begin{pmatrix}
e^{\im g_\uparrow(x)}u_\uparrow(x)\\
e^{\im g_\downarrow(x)}u_\downarrow(x)
\end{pmatrix},
\ \text{where }g_\uparrow(x),\ g_\downarrow(x)\in \R\ \text{for all }x\in\Z. 
\end{align*}
A gauge transformation is obviously a unitary operator on $\mathcal H$.

Following \cite{FO17JFA}, we can show that there exists a gauge transformation that reduces the coin to the case $\beta(x)>0$.

\begin{lemma}\label{lem:Gauge}
	There exists a gauge transformation $G$ such that
	\begin{align*}
	U=G^{-1}S C_{\alpha(x),\sqrt{1-|\alpha(x)|^2},\theta(x)} G.
	\end{align*}
\end{lemma}

\begin{proof}
	We set $\beta(x)=e^{\im b(x)}|\beta(x)|$ with $b(x)\in [0,2\pi)$.
	We set
	\begin{align*}
	B(x)=\begin{cases} \sum_{y=0}^{x-1}b(y) & x\geq 1\\ 0 & x=0\\ -\sum_{y=x}^{-1}b(y) & x\leq -1.\end{cases}
	\end{align*}
	Then, we have
	\begin{align}\label{rel:Bb}
	B(x+ 1)-B(x)=b(x).
	\end{align}
	Using this $B$, we set $G$ by $g_\uparrow(x)=-B(x)$ and $g_\downarrow(x)=-B(x+1)$.
	By using \eqref{rel:Bb}, we have
	\begin{align*}
	\(GS C_{\alpha, \beta, \theta}u\)(x)&=\begin{pmatrix} e^{-\im B(x)} e^{\im \theta(x-1)}\(\beta(x-1) u_{\uparrow}(x-1)+\overline{\alpha}(x-1)  u_{\downarrow}(x-1)\) \\ e^{-\im B(x+1)}e^{\im \theta(x+1)}\(-\alpha(x+1) u_{\uparrow}(x+1)+\overline{\beta}(x+1) u_{\downarrow}(x+1)\)\end{pmatrix}\\&
	=\begin{pmatrix}  e^{\im \theta(x-1)}\(|\beta(x-1)|e^{-\im B(x-1)} u_{\uparrow}(x-1)+\overline{\alpha}(x-1)e^{-\im B(x)} u_{\downarrow}(x-1)\) \\ e^{\im \theta(x+1)}\(-\alpha(x+1)e^{-\im B(x+1)} u_{\uparrow}(x+1)+|\beta(x+1)|e^{-\im B(x+2)} u_{\downarrow}(x+1)\)\end{pmatrix}\\&
	=S C_{\alpha,|\beta|, \theta} Gu (x).
	\end{align*}
	Therefore, we have the conclusion.
\end{proof}

Notice that the gauge transformation given in Lemma \ref{lem:Gauge} do not change $\alpha(x)$ and $\theta(x)$.
Therefore, the assumptions in Assumption \ref{ass:1}, Propositions \ref{prop:spec}, \ref{prop:finite} and Theorem \ref{thm:main} do not change.
Moreover, since gauge transformation is a unitary operator, the spectral properties stated in Propositions \ref{prop:spec} and \ref{prop:finite} do not change.
The conclusion of Theorem \ref{thm:main} also holds with exactly the same (implicit) constant because $\|Gu\|_{l^{p,\sigma}}=\|u\|_{l^{p,\sigma}}$.

Therefore, in the following, it suffices to consider the case $\beta(x)=\rho(x):=\sqrt{1-|\alpha(x)|^2}$.

\section{Reformulation of generalized eigenvalue problem}\label{sec:2}

In this section, we study the following (generalized) eigenvalue problem:
\begin{align}\label{ev:prob}
U u = e^{\im \lambda}u,\quad \lambda\in \C/2\pi\Z,
\end{align}
where $u:\Z\to \C^2$ is not necessary in $\mathcal H$.

We will frequently use the Pauli matrices:
\begin{align}\label{pauli}
\sigma_0:=\begin{pmatrix}1&0\\0&1\end{pmatrix},\  \sigma_1:=\begin{pmatrix}0 &1\\1&0\end{pmatrix},\  \sigma_2:=\begin{pmatrix}0&-\im \\ \im &0\end{pmatrix},\  \sigma_3:=\begin{pmatrix}1 &0\\0&-1\end{pmatrix}.
\end{align}
We set a unitary operator $J_{\mathrm{VE}}$ by
\begin{align*}
J_{\mathrm{VE}} u(x):=\begin{pmatrix} u_{\downarrow}(x-1) \\ u_{\uparrow}(x) \end{pmatrix}.
\end{align*}
\begin{remark}
	The operator $J_{\mathrm{VE}}$ express a shift of view point from the vertex to the edge of the graph $\Z$.
\end{remark}
We denote the inverse of $J_{\mathrm{VE}}$ by 
\begin{align*}
J_{\mathrm{EV}} u(x):=J_{\mathrm{VE}}^{-1} u(x)=\begin{pmatrix}u_{\downarrow}(x) \\ u_{\uparrow}(x+1)\end{pmatrix}
\end{align*}
For $\lambda\in  \C/2\pi\Z$ and $x\in\Z$, we set
\begin{align*}
T_\lambda(x):=\rho(x)^{-1}\begin{pmatrix}  e^{\im (\lambda-\theta(x))} & \alpha(x) \\ \overline{\alpha(x)} & e^{-\im (\lambda-\theta(x))} \end{pmatrix}.
\end{align*}
Notice that $T_\lambda(x)$ has the following symmetry:
\begin{align}\label{16}
T_{\bar \lambda}(x)=\overline{\sigma_1 T_{ \lambda}(x)\sigma_1}.
\end{align}

The following equivalence is the crux of our analysis which heavily relies on the fact that the base space $\Z$ is one dimensional.

\begin{proposition}\label{prop:equiv}
	For $u:\Z\to \C^2$, the generalized eigenvalue problem \eqref{ev:prob} is equivalent to
	\begin{align}\label{eq:equiv}
	\(J_{\mathrm{VE}} u\)(x+1)=T_\lambda(x)\(J_{\mathrm{VE}} u\)(x),\ \forall x\in\Z.
	\end{align}
	
\end{proposition}

\begin{proof}
	Writing down the equation $Uu=e^{\im \lambda}u$ explicitly, we have
	\begin{equation}\label{evprobel}
	\begin{aligned}
	e^{\im \theta(x-1)}\rho (x-1)u_{\uparrow}(x-1)+ e^{\im  \theta(x-1)}  \overline{\alpha}(x-1)u_{\downarrow}(x-1)&=e^{\im  \lambda}u_{\uparrow}(x),\\
	-e^{\im  \theta(x+1)} \alpha (x+1)u_{\uparrow}(x+1)+e^{\im  \theta(x+1)}  \rho (x+1)u_{\downarrow}(x+1)&=e^{\im  \lambda}u_{\downarrow}(x).
	\end{aligned}
	\end{equation}
	Rearranging \eqref{evprobel}, we have
	\begin{align}\label{evprobel2}
	\begin{pmatrix} 0 & e^{\im  \theta(x)}\rho (x)\\ -e^{\im  \lambda} & -e^{\im  \theta(x)} \alpha (x)\end{pmatrix} 
	\begin{pmatrix} u_{\downarrow}(x-1) \\ u_{\uparrow} (x)\end{pmatrix} = \begin{pmatrix} -e^{\im \theta(x)}  \overline{\alpha}(x)& e^{\im \lambda} \\ -e^{\im \theta(x)} \rho (x) & 0 \end{pmatrix}
	\begin{pmatrix} u_{\downarrow}(x) \\ u_{\uparrow} (x+1)\end{pmatrix}.
	\end{align}
	Thus, we obtain
	\begin{align*}
	\begin{pmatrix} u_{\downarrow}(x) \\ u_{\uparrow} (x+1)\end{pmatrix}=  \rho (x)^{-1}\begin{pmatrix} e^{\im ( \lambda-\theta(x))} &  \alpha (x)\\  \overline{\alpha}(x) &  e^{-\im (\lambda-\theta(x))}\end{pmatrix}
	\begin{pmatrix} u_{\downarrow}(x-1) \\ u_{\uparrow} (x)\end{pmatrix}.
	\end{align*}
	This is the explicit form of \eqref{eq:equiv}. Since all the manipulations, we have done can be reversed, we have the conclusion.
\end{proof}

As ordinary differential equations, we can show the Wronskian does not depend on $x\in \Z$.

\begin{proposition}\label{prop:Wronsky}
	Let $v_1=J_{\mathrm{VE}} u_1$ and $\ v_2=J_{\mathrm{VE}} u_2$ satisfy \eqref{eq:equiv}.
	Then, $\det (v_1(x)\ v_2(x))$ is independent of $x\in\Z$.
\end{proposition}

\begin{proof}
	Since $\det T_\lambda(x)=1$, we have
	\begin{align*}
	\det (v_1(x+1)\ v_2(x+1))&=\det (T_\lambda(x)v_1(x)\ T_\lambda(x)v_2(x)) =\det T_\lambda(x)\det (v_1(x)\ v_2(x))\\&=\det (v_1(x)\ v_2(x)).
	\end{align*}
	Therefore, we have the conclusion.
\end{proof}

An immediate corollary of Proposition \ref{prop:Wronsky} is the following.
\begin{corollary}\label{cor:Wron}
	Let $u_1$ and $u_2$ satisfy \eqref{ev:prob}, $u_1$ is bounded on $\Z_{\geq 0}:=\{x\in\Z\ |\ x\geq 0\}$ and $u_2\in \mathcal H$.
	Then, $u_1$ and $u_2$ are linearly dependent.
\end{corollary}

\begin{proof}
	Since $J_{\mathrm{VE}} u_2\in \mathcal{H}=l^2(\Z,\C^2)$, there exists a sequence $\{x_{n}\}_{n=1}^{\infty}\subset \Z_{\ge 0}$ such that $x_n\to \infty$ and \ $\|\(J_{\mathrm{VE}} u_2\)(x_n)\|\to 0$ as $n\to \infty$.
	Thus, we have
	\begin{align*}
	\det(\(J_{\mathrm{VE}} u_1\)(x), \(J_{\mathrm{VE}} u_2\)(x))=\lim_{x\to \infty}\det(\(J_{\mathrm{VE}} u_1\)(x_n), \(J_{\mathrm{VE}} u_2\)(x_n))=0.
	\end{align*}
	Therefore, we have the conclusion.
\end{proof}

As was shown in the above, for our analysis it is more convenient to consider $J_{\mathrm{VE}}u$ rather than $u$ itself.
Therefore, we set  
\begin{align}\label{def:CMV}
\mathcal C=J_{\mathrm{VE}}UJ_{\mathrm{EV}}.
\end{align}
\begin{remark}
	The operator $\mathcal C$ is a CMV matrix, see \cite{CGMV10CPAM}.
	However, we will perform all computations directly and not use the deep theories of orthogonal polynomials.
\end{remark}

For $v\in \C^2$ (column vector), we set $v^\top$ (row vector) to be the transposition of $v$.
We can express the kernel of the resolvent of $U$ by using the generalized eigenfunctions.
The following Proposition is a special case of Lemma 3.1 of \cite{GZ06JAT}.

\begin{proposition}\label{prop:kernel}
	Let $v_1=J_{\mathrm{VE}} u_1$ and $v_2=J_{\mathrm{VE}} u_2$ satisfy \eqref{eq:equiv} with $W_\lambda=\det(v_1\ v_2)\neq 0$.
	We set
	\begin{align*}
	K_\lambda(x,y):=e^{-\im \lambda} W_\lambda^{-1} \(v_2(x) v_1(y)^\top \begin{pmatrix} 0 & 1_{<y}(x)\\ 1_{\leq y}(x) & 0 \end{pmatrix} + v_1(x) v_2(y)^\top \begin{pmatrix} 0 & 1_{\geq y}(x)\\ 1_{> y}(x) & 0 \end{pmatrix} \).
	\end{align*}
	Then, we have
	\begin{align*}
	\( (\mathcal C-e^{\im \lambda})  K_\lambda(\cdot,y)\)(x)
	=\begin{cases} \sigma_0 & x=y,\\ 0 & x\neq y. \end{cases},
	\end{align*}
	where $\sigma_0$ is given in \eqref{pauli}.
\end{proposition}

Before proving Proposition \ref{prop:kernel}, we give the expression of $\mathcal{C}$. 

\begin{lemma}\label{lem:cmv}
	We have
	\begin{align*}
	\(\mathcal{C} v\)(x)=V_{x,x-1}v(x-1)+V_{x,x}v(x)+V_{x,x+1}v(x+1),
	\end{align*}
	where
	\begin{align*}
	V_{x,x-1}&=\begin{pmatrix} 0 & 0 \\ 0 & e^{\im \theta(x-1)}\rho(x-1) \end{pmatrix},\ 
	V_{x,x}=\begin{pmatrix} 0 & -e^{\im \theta(x)}\alpha(x) \\ e^{\im \theta(x-1)}\overline{\alpha(x-1)} & 0 \end{pmatrix} ,\\ 
	V_{x,x+1}&=\begin{pmatrix} e^{\im \theta(x)}\rho(x) & 0 \\ 0 & 0 \end{pmatrix} .
	\end{align*}
\end{lemma}

\begin{proof}
	By
	\begin{align*}
	&\(J_{\mathrm{VE}} U J_{\mathrm{EV}} v\)(x)=\begin{pmatrix} \(UJ_{\mathrm{EV}} v\)_{\downarrow }(x-1) \\ \(UJ_{\mathrm{EV}} v\)_{\uparrow}(x)\end{pmatrix}\\&=
	\begin{pmatrix} e^{\im \theta(x)}\(-\alpha(x)\(J_{\mathrm{EV}} v\)_{\uparrow}(x)+\rho(x)\(J_{\mathrm{EV}} v\)_{\downarrow }(x)\) \\ e^{\im \theta(x-1)}\(\rho(x-1)\(J_{\mathrm{EV}} v\)_{\uparrow}(x-1)+\overline{\alpha(x-1)}\(J_{\mathrm{EV}} v\)_{\downarrow}(x-1)\)\end{pmatrix}
	\\&=
	\begin{pmatrix} e^{\im \theta(x)}\(-\alpha(x)v_{\downarrow}(x)+\rho(x)v_{\uparrow}(x+1)\) \\ e^{\im \theta(x-1)}\(\rho(x-1)v_{\downarrow}(x-1)+\overline{\alpha(x-1)}v_{\uparrow}(x)\)\end{pmatrix}
	\\&=
	\begin{pmatrix} 0 & 0 \\ 0 & e^{\im \theta(x-1)}\rho(x-1) \end{pmatrix} \begin{pmatrix} v_{\uparrow}(x-1) \\ v_{\downarrow}(x-1) \end{pmatrix} 
	+
	\begin{pmatrix} 0 & -e^{\im \theta(x)}\alpha(x) \\ e^{\im \theta(x-1)}\overline{\alpha(x-1)} & 0 \end{pmatrix} \begin{pmatrix} v_{\uparrow}(x) \\ v_{\downarrow}(x) \end{pmatrix} \\&\quad
	+
	\begin{pmatrix} e^{\im \theta(x)}\rho(x) & 0 \\ 0 & 0 \end{pmatrix} \begin{pmatrix} v_{\uparrow}(x+1) \\ v_{\downarrow}(x+1) \end{pmatrix}. 
	\end{align*}
	Therefore, we have the conclusion.
\end{proof}


\begin{proof}[Proof of Proposition \ref{prop:kernel}]
	Since $v_1,v_2$ solve \eqref{eq:equiv}, from \eqref{ev:prob}, Proposition \ref{prop:equiv} and Lemma \ref{lem:cmv}, we have
	\begin{align}\label{ker:0}
	V_{x,x-1}v_j(x-1)+(V_{x,x}-e^{\im \lambda})v_j(x)+V_{x,x+1}v_j(x+1)=0,\ \forall x\in \Z,\ j=1,2.
	\end{align}
	By Lemma \ref{lem:cmv}, we have
	\begin{equation*}
	\( (\mathcal{C}-e^{\im \lambda})  K_\lambda(\cdot,y)\)(x)=V_{x,x-1}K_\lambda(x-1,y)+\(V_{x,x}-e^{\im \lambda}\)K_\lambda(x,y)+V_{x,x+1}K_\lambda(x+1,y).
	\end{equation*}
	For $x>y+1$, we have
	\begin{align*}
	&\( (\mathcal{C}-e^{\im \lambda})  K_\lambda(\cdot,y)\)(x)\\&=e^{-\im\lambda}W_{\lambda}^{-1}\(V_{x,x-1}v_1(x-1)+(V_{x,x}-e^{\im \lambda})v_1(x)+V_{x,x+1}v_1(x+1)\)v_2(y)^\top \sigma_1=0,
	\end{align*}
	where the last equality follows from \eqref{ker:0}.
	Similarly, for $x<y-1$, we have
	\begin{align*}
	\((\mathcal{C}-e^{\im \lambda}) K_\lambda(\cdot,y)\)(x)=0.
	\end{align*}
	For the case $x=y+1$, we have
	\begin{align*}
	&\((\mathcal{C}-e^{\im \lambda}) K_\lambda(\cdot,y)\)(y+1)\\&=V_{y+1,y}K_\lambda(y,y)+\(V_{y+1,y+1}-e^{\im \lambda}\)K_\lambda(y+1,y)+V_{y+1,y+2}K_\lambda(y+2,y)\\&=
	e^{-\im\lambda}W_{\lambda}^{-1}V_{y+1,y}\(v_2(y)v_1(y)^\top\begin{pmatrix}0 & 0 \\ 1 & 0 \end{pmatrix}+v_1(y)v_2(y)^\top \begin{pmatrix}0 & 1 \\ 0 & 0 \end{pmatrix}\)\\&\quad+e^{-\im\lambda}W_{\lambda}^{-1}\(\(V_{y+1,y+1}-e^{\im \lambda}\)v_1(y+1) +V_{y+1,y+2}v_1(y+2)\)v_2(y)^\top \sigma_1\\&=e^{-\im\lambda}W_{\lambda}^{-1}V_{y+1,y}\(v_2(y)v_1(y)^\top-v_1(y)v_2(y)^\top\)\begin{pmatrix}0 & 0 \\ 1 & 0 \end{pmatrix}
	\\&=
	e^{-\im\lambda}W_{\lambda}^{-1}\begin{pmatrix} 0 & 0 \\ 0 & e^{\im \theta(y)}\rho(y) \end{pmatrix}\\&
	\quad \times \begin{pmatrix} 0 & v_{1,\downarrow }(y)v_{2,\uparrow}(y)-v_{1,\uparrow}(y)v_{2,\downarrow }(y)\\ v_{1,\uparrow}(y)v_{2,\downarrow }(y)-v_{1,\downarrow }(y)v_{2,\uparrow}(y)& 0 \end{pmatrix}\begin{pmatrix}0 & 0 \\ 1 & 0 \end{pmatrix}\\&=0,
	\end{align*}
	where we have used \eqref{ker:0} in the 3rd equality.
	Similarly, for the case $x=y-1$, we have
	\begin{align*}
	&\((\mathcal{C}-e^{\im \lambda}) K_\lambda(\cdot,y)\)(y-1)=0.
	\end{align*}
	
	Finally, we consider the case $x=y$.
	As before,
	\begin{align*}
	&\((\mathcal{C}-e^{\im \lambda}) K_\lambda(\cdot,y)\)(y)=V_{y,y-1}K_\lambda(y-1,y)+\(V_{y,y}-e^{\im \lambda}\)K_\lambda(y,y)+V_{y,y+1}K_\lambda(y+1,y)\\&=
	e^{-\im\lambda}W_{\lambda}^{-1}V_{y,y-1}v_2(y-1)v_1(y)^\top \sigma_1
	\\&\quad+e^{-\im\lambda}W_{\lambda}^{-1}\(V_{y,y}-e^{\im \lambda}\)\(v_2(y)v_1(y)^\top \begin{pmatrix} 0 & 0 \\ 1 & 0\end{pmatrix} 
	+ v_1(y)v_2(y)^\top \begin{pmatrix} 0 & 1 \\ 0 & 0\end{pmatrix}\)\\&\quad+e^{-\im\lambda}W_{\lambda}^{-1}V_{y,y+1}v_1(y+1)v_2(y)^\top \sigma_1.
	\end{align*}
	Then, 
	\begin{align*}
	&V_{y,y-1}v_2(y-1)v_1(y)^\top \sigma_1+V_{y,y+1}v_1(y+1)v_2(y)^\top \sigma_1\\&=
	\begin{pmatrix} 0 & 0 \\ 0 & e^{\im \theta(y-1)}\rho(y-1)\end{pmatrix}\begin{pmatrix} v_{1,\uparrow}(y)v_{2,\uparrow}(y-1)& v_{1,\downarrow }(y)v_{2,\uparrow}(y-1)\\ v_{1,\uparrow}(y)v_{2,\downarrow }(y-1) & v_{1,\downarrow }(y)v_{2,\downarrow }(y-1) \end{pmatrix}\sigma_1\\&\quad
	+
	\begin{pmatrix} e^{\im \theta(y )}\rho(y )   & 0 \\ 0 & 0 \end{pmatrix}\begin{pmatrix}v_{1,\uparrow}(y+1)v_{2,\uparrow}(y)& v_{1,\uparrow}(y+1)v_{2,\downarrow }(y)\\ v_{1,\downarrow }(y+1)v_{2,\uparrow}(y) & v_{1,\downarrow }(y+1)v_{2,\downarrow }(y)\end{pmatrix}\sigma_1\\&
	=
	\begin{pmatrix}
	e^{\im \theta(y )}\rho(y )v_{1,\uparrow}(y+1)v_{2,\downarrow }(y) & e^{\im \theta(y )}\rho(y )v_{1,\uparrow}(y+1)v_{2,\uparrow}(y)\\
	e^{\im \theta(y-1)}\rho(y-1)v_{1,\downarrow }(y)v_{2,\downarrow }(y-1) & e^{\im \theta(y-1)}\rho(y-1)v_{1,\uparrow}(y)v_{2,\downarrow }(y-1)
	\end{pmatrix},
	\end{align*}
	and
	\begin{align*}
	&\(V_{y,y}-e^{\im \lambda}\)\(v_2(y)v_1(y)^\top \begin{pmatrix} 0 & 0 \\ 1 & 0\end{pmatrix} + v_1(y)v_2(y)^\top \begin{pmatrix} 0 & 1 \\ 0 & 0\end{pmatrix}\)\\&
	=\begin{pmatrix} -e^{\im \lambda} & -e^{\im \theta(y)}\alpha(y) \\ e^{\im \theta(y-1)}\overline{\alpha(y-1)}  & -e^{\im \lambda} \end{pmatrix}
	\begin{pmatrix}
	v_{1,\downarrow }(y)v_{2,\uparrow}(y) & v_{1,\uparrow}(y)v_{2,\uparrow}(y) \\
	v_{1,\downarrow }(y)v_{2,\downarrow }(y) & v_{1,\downarrow }(y)v_{2,\uparrow}(y)
	\end{pmatrix}.
	\end{align*}
	Thus, setting $A_{i,j}:=e^{\im \lambda}W_\lambda\(\((\mathcal{C}-e^{\im \lambda}) K_\lambda(\cdot,y)\)(y)\)_{i,j}$ ($i,j=1,2$), we have
	\begin{align*}
	A_{1,1}&=e^{\im \theta(y )}\rho(y )v_{1,\uparrow}(y+1)v_{2,\downarrow }(y) -e^{\im \lambda}v_{1,\downarrow }(y)v_{2,\uparrow}(y)-e^{\im \theta(y)}\alpha(y)v_{1,\downarrow }(y)v_{2,\downarrow }(y), 
	\\
	A_{1,2}&=e^{\im \theta(y )}\rho(y )v_{1,\uparrow}(y+1)v_{2,\uparrow}(y)-e^{\im \lambda} v_{1,\uparrow}(y)v_{2,\uparrow}(y) -e^{\im \theta(y)}\alpha(y)v_{1,\downarrow }(y)v_{2,\uparrow}(y),
	\\
	A_{2,1}&=e^{\im \theta(y-1)}\rho(y-1)v_{1,\downarrow }(y)v_{2,\downarrow }(y-1) +e^{\im \theta(y-1)}\overline{\alpha(y-1)}v_{1,\downarrow }(y)v_{2,\uparrow}(y)-e^{\im \lambda}v_{1,\downarrow }(y)v_{2,\downarrow }(y),
	\\
	A_{2,2}&= e^{\im \theta(y-1)}\rho(y-1)v_{1,\uparrow}(y)v_{2,\downarrow }(y-1) + e^{\im \theta(y-1)}\overline{\alpha(y-1)}v_{1,\uparrow}(y)v_{2,\uparrow}(y)-e^{\im \lambda}v_{1,\downarrow }(y)v_{2,\uparrow}(y).
	\end{align*}
	We will show that $A_{1,2}=A_{2,1}=0$ and $A_{1,1}=A_{2,2}=e^{\im \lambda}W_\lambda$.
	Rewriting \eqref{evprobel2} in terms of $v_j$, we have
	\begin{align}
	e^{\im \theta(x)}\rho(x) v_{j,\downarrow }(x)=-e^{\im \theta(x)}\overline{\alpha(x)}v_{j,\uparrow}(x+1)+e^{\im \lambda}v_{j,\downarrow }(x+1),\label{eq:v1}\\
	-e^{\im \lambda}v_{j,\uparrow}(x)-e^{\im \theta(x)}\alpha(x)v_{j,\downarrow }(x)=-e^{\im \theta(x)}\rho(x)v_{j,\uparrow}(x+1).\label{eq:v2}
	\end{align}
	Then, by \eqref{eq:v2} with $x=y$ and $j=1$, we have
	\begin{align*}
	A_{1,2}=\(e^{\im \theta(y )}\rho(y )v_{1,\uparrow}(y+1)-e^{\im \lambda} v_{1,\uparrow}(y) -e^{\im \theta(y)}\alpha(y)v_{1,\downarrow }(y)\)v_{2,\uparrow}(y)=0,
	\end{align*}
	and by \eqref{eq:v1} with $x=y-1$ and $j=2$, we have
	\begin{align*}
	A_{2,1}&=\(e^{\im \theta(y-1)}\rho(y-1)v_{2,\downarrow }(y-1) +e^{\im \theta(y-1)}\overline{\alpha(y-1)}v_{2,\uparrow}(y)-e^{\im \lambda}v_{2,\downarrow }(y)\)v_{1,\downarrow }(y)=0.
	\end{align*}
	For the diagonal terms, by \eqref{eq:v2} with $x=y$ and $j=1$, we have
	\begin{align*}
	A_{1,1}&=\(e^{\im \theta(y )}\rho(y )v_{1,\uparrow}(y+1)-e^{\im \theta(y)}\alpha(y)v_{1,\downarrow }(y)\)v_{2,\downarrow }(y) -e^{\im \lambda}v_{1,\downarrow }(y)v_{2,\uparrow}(y)\\&
	=e^{\im \lambda}\(v_{1,\uparrow}(y)v_{2,\downarrow }(y) -v_{1,\downarrow }(y)v_{2,\uparrow}(y)\)=e^{\im \lambda}W_\lambda.
	\end{align*}
	and by \eqref{eq:v1} with $x=y-1$ and $j=2$, we have
	\begin{align*}
	A_{2,2}&=\(e^{\im \theta(y-1)}\rho(y-1) v_{2,\downarrow }(y-1) + e^{\im \theta(y-1)}\overline{\alpha(y-1)}v_{2,\uparrow}(y)\)v_{1,\uparrow}(y)-e^{\im \lambda}v_{1,\downarrow }(y)v_{2,\uparrow}(y)\\&
	=e^{\im \lambda}\(v_{1,\uparrow}(y)v_{2,\downarrow }(y) -v_{1,\downarrow }(y)v_{2,\uparrow}(y)\)=e^{\im \lambda}W_\lambda.
	\end{align*}
	Therefore, the conclusion follows.
\end{proof}

\section{The unperturbed case}\label{sec:3}

We consider $e^{\im z}$ and $\cos z$ as an function on $\Cp=\C/2\pi\Z$ instead of $\C$.
Notice that $-z$, $\bar z$ are well-defined in $\Cp$.
Let $\T :=\{z\in \Cp\ |\ \Im z=0\}$ and $\Tp:=\{z\in \Cp\ |\ \Im z>0\}$.

In this section, we study the case $\alpha(x)\equiv -\overline{\alpha_0}$ and $\theta(x)\equiv 0$.
We denote the corresponding coin operator by $\hat C_0$ and set $U_0=S\hat C_0$.
As \eqref{def:CMV}, we also set $\mathcal{C}_0:=J_{\mathrm{VE}}U_0 J_{\mathrm{EV}}$.
The transfer matrix is $$T_\lambda (x)=T_{0,\lambda} :=\rho_0^{-1}\begin{pmatrix} e^{\im \lambda} &-\overline{ \alpha_0}  \\   -\alpha_0 & e^{-\im \lambda}\end{pmatrix}.$$
Since $\det T_{0,\lambda} =1$, the two eigenvalues of $T_{0,\lambda}$ can be expressed as $e^{\pm\im \xi(\lambda)}$.
In particular, by $\mathrm{tr}T_{0,\lambda} =\rho_0^{-1}(e^{\im \lambda}+e^{-\im \lambda})=2\rho_0^{-1}\cos \lambda$, we obtain the relation
\begin{align}
\cos \xi= \rho_0^{-1}\cos \lambda.\label{21}
\end{align}
Conversely, if $(\xi,\lambda)$ satisfy \eqref{21}, $e^{\pm\im  \xi}$ are eigenvalues of $T_{0,\lambda}$. 
Further, we set the closed interval $I_n\subset \T $ ($n=1,2$) by
\begin{align}\label{21.2}
I_1:=\{\xi\in [0,\pi]\ |\ |\cos \xi|\leq \rho_0\},\quad I_2:=-I_1=\{\xi\in [-\pi,0]\ |\ -\xi\in I_1\},
\end{align}
and set $I_0=I_1\cup I_2$.

\begin{proposition}\label{prop:4.1}
For given $\lambda\in \Cp\setminus I_0$, there exists a unique $\xi=\xi(\lambda)\in \T_+$ which satisfies \eqref{21}.
Further, the map $\xi:\Cp\setminus I_0 \to \Tp$ is holomorphic.
\end{proposition}
For the proof we use the following properties of the function $\cos$.
For the convenience of readers, we give the proof in the appendix of this paper (Section \ref{proof:tri}).
\begin{lemma}\label{lem:4.2}
The trigonometric function $\cos:\Tp \to \C\setminus [-1,1]$ is a biholomorphism. 
We set $\cos^{-1}:\C\setminus [-1,1]\to \Tp$.
Moreover, we can analytically extend $\cos^{-1}$ through the cut $(-1,1)$ defining $\cos^{-1}_\pm:\{z\in \C\ |\ |\Re z|<1\}\to \C$ by
\begin{align*}
\cos^{-1}_+ z:=\begin{cases} \cos^{-1} z & \Re z\in (-1,1),\ \Im z\geq 0,\\
\overline{\cos^{-1}\bar z}&\Re z\in (-1,1),\ \Im z< 0,
\end{cases}
\\
\cos^{-1}_- z:=\begin{cases} \overline{\cos^{-1}\bar z}  & \Re z\in (-1,1),\ \Im z> 0,\\
   \cos^{-1} z&\Re z\in (-1,1),\ \Im z\leq 0.
\end{cases}
\end{align*}
In particular
\begin{align}\label{21.5}
\cos^{-1}_\pm(\xi)=\mp\arccos \xi,\quad \xi \in \T.
\end{align}
Here, $\arccos$ is the inverse of $\cos:(0,\pi)\to (-1,1)$.
\end{lemma}

\begin{proof}[Proof of Proposition \ref{prop:4.1}]
Let $\lambda\in \Cp \setminus I_0$.
Then, we have $\rho_0^{-1} \cos \lambda \in \C\setminus [-1,1]$.
Therefore, setting $\cos^{-1}$ to be the inverse of $\cos:\Tp \to \C\setminus[-1,1]$, we can set
\begin{align*}
\xi(\lambda)=\cos^{-1} \(\rho_0^{-1}\cos \lambda\),
\end{align*}
which is obviously holomorphic and satisfies \eqref{21}.
Since $\cos:\Tp \to \C\setminus[-1,1]$ is an injection, the uniqueness also follows.
\end{proof}

Even though the map $\lambda\mapsto \xi(\lambda)$ seems to be more natural than the inverse map $\xi\mapsto \lambda(\xi)$ (which is not defined at this moment), we need it for the dispersive theory.
\begin{remark}
For the Schr\"odinger equation, it is $\xi\mapsto \lambda(\xi)=\xi^2$ and for the discrete Schr\"odinnger equation it is $\xi\mapsto \lambda(\xi)= 2-2\cos \xi$.
However, in our case we have to handle the inverse of $\cos$.
\end{remark}


Set $\delta_0>0$ to be the unique solution of $\rho_0\cosh \delta_0=1$.
Set $\T_\delta:=\{z\in\Cp\ |\ |\Im z|<\delta \}$.
We define holomorphic functions $\lambda_\pm$ on $\T_{\delta_0}$ by
\begin{align}\label{21.6}
\lambda_\pm(\xi)=\cos^{-1}_\pm\(\rho_0\cos \xi\).
\end{align}
Notice that since $\cos(-\xi)=\cos \xi$, we have
\begin{align}\label{21.7}
\lambda_\pm(-\xi)=\lambda_\pm(\xi).
\end{align}

\begin{remark}
Notice $|\rho_0\Re \cos\xi|=|\rho_0\cosh (\Im \xi)\cos (\Re\xi)|<1$, for $\xi\in \T_{\delta_0}$.
Therefore, we have $\rho_0\cos \T_{\delta_0}\subset \{z\in \C\ |\ |\Re z|<1\}$ and \eqref{21.6} give well defined holomorphic functions.
\end{remark}

By \eqref{21.5} we have
\begin{align*}
\left.\lambda_\pm\right|_\T\in C^\omega(\T;\R),\ \left.\lambda_\pm\right|_\T (\xi)= \mp \arccos \(\rho_0\cos\xi\).
\end{align*}

\begin{remark}
If $\xi \in \T$, then $\lambda_\pm(\xi)\in I_{\frac 3 2 \pm \frac12}$.
Thus, the two arcs $\{e^{\im \lambda_\pm(\xi)}\ |\ \xi \in \T\}$  correspond to $\sigma_{\mathrm{ess}}(\cC)$ and in particular, the spectrum of $\sigma(\cC_0)$.
Further, if $\delta>0$, $\{\xi\in \T\ |\ \lambda_\pm(\xi+\im \delta)\}$ will encircle $I_{\frac 3 2 \pm \frac12}$.
\end{remark}

In the following, we will consider only $\lambda_-$ which corresponds to the spectrum of $\cC$ in the upper half plain.
We simply denote $\lambda_-(\xi)$ by $\lambda(\xi)$.

Let $\xi\in \T_{\delta}$ where $\delta>0$ sufficiently small so that
\begin{align}\label{23.1.1}
A_\pm(\xi):=\(|\alpha_0|^2+(\sin \lambda(\xi) \mp \rho_0 \sin \xi)^2\)^{1/2}\neq 0\text{ for all } \xi \in \T_\delta.
\end{align}
\begin{remark}
If $\xi\in \T$, we have $\(|\alpha_0|^2+(\sin \lambda(\xi) \mp \rho_0 \sin \xi)^2\)\geq |\alpha_0|^2$ (notice that it is real), so there exists $\delta>0$ s.t.\  if $\xi\in \T_{\delta}$, then $$\Re\(|\alpha_0|^2+(\sin \lambda(\xi) \mp \rho_0 \sin \xi)^2\)>|\alpha_0|/2,$$
by continuity.
Here, the root in \label{23.1.1} is taken in the usual manner $(re^{\im \theta})^{1/2}=r^{1/2}e^{\im \theta/2}$ for $\theta\in (-\pi,\pi]$.
Also, for $\xi \in \T$, since $\lambda(\xi)\in \R$, we have
\end{remark}
We set 
\begin{align}\label{23.2}
\varphi_{\pm}(\xi)=A_\pm(\xi)^{-1}\begin{pmatrix}\overline{\alpha_0}\\ e^{\im \lambda(\xi)}-\rho_0 e^{\pm\im \xi}\end{pmatrix}.
\end{align}
On can easily show that $\varphi_\pm(\xi)$ are eigenvectors of $T_{0,\lambda(\xi)}$ associated to the eigenvalues $e^{\pm \im \xi}$.
Further, by \eqref{21}, for $\xi\in \T$, since $e^{\im \lambda(\xi)}- \rho_0e^{\pm \im \xi}=\im \(\sin \lambda(\xi)\mp \rho_0 \sin \xi\)$ so we have $|\alpha_0|^2+(\sin \lambda(\xi) \mp \rho_0 \sin \xi)^2=|\alpha_0|^2+|e^{\im \lambda(\xi)}-\rho_0 e^{\pm\im \xi}|^2$.

\begin{lemma}\label{lem:6}
$\T_{\delta}\ni \xi\mapsto \varphi_\pm(\xi)\in \C^2$ is holomorphic.
Moreover, we have $\varphi_\pm(-\xi)=\varphi_\mp(\xi)$  and $\|\varphi_\pm(\xi)\|_{\C^2}=1$ for $\xi \in \T$.
\end{lemma}

\begin{proof}
The factor $A_\pm(\xi)^{-1}$ is multiplied to normalize $\varphi_\pm(\xi)$ at $\xi \in \T$ (only at $\xi \in \T$ and we note that it is not normalized in $\xi \in \T_\delta\setminus\T$).
Notice that if we multiply $(|\alpha_0|^2+|e^{\im \lambda(\xi)}-\rho_0 e^{\pm \im \xi}|^2)^{-1}$ instead of $A_\pm(\xi)^{-1}$, $\varphi_\pm(\xi)$ will be normalized in $\xi\in \T_\delta$ but will not be holomorphic any more.

The second claim can be seen from \eqref{21.7} and the explicit form of $ \varphi_\pm (\xi)$ given in \eqref{23.2}.
\end{proof}

From the symmetry \eqref{16}, we have the following identity.
\begin{lemma}\label{lem:6.3}
There exists $\gamma\in C^\omega(\T;\C)$ s.t.\ $|\gamma(\xi)|=1$ and
\begin{align}\label{23.4}
\gamma(\xi)\sigma_1\overline{\varphi_\pm(\xi)}= \varphi_\pm(-\xi)
\end{align}
\end{lemma}

\begin{proof}
Taking the complex conjugate and multiplying $\sigma_1$ to $T_{\lambda(\xi)} \varphi_\pm(\xi)=e^{\pm \im \xi}\varphi_\pm(\xi)$, we have
\begin{align*}
\sigma_1 \overline{T_{\lambda(\xi)}}\sigma_1 \sigma_1 \overline{\varphi_\pm (\xi)}= e^{\mp \im \xi} \sigma_1 \overline{\varphi_\pm(\xi)}.
\end{align*}
Recall that for $\xi \in \T$, $\lambda(\xi)\in \R$ and $\|\varphi_\pm(\xi)\|_{\C^2}=1$.
Thus, by \eqref{16}, there exists $\gamma_{\pm,\xi}$ s.t.\ $|\gamma_{\pm,\xi}|=1$ and $\gamma_{\pm,\xi}\sigma_1\overline{\varphi_\pm(\xi)}= \varphi_\pm(-\xi)$.
The remaining task is to show $\gamma_{+,\xi}=\gamma_{-,\xi}$.
However, this is easily deduced from
\begin{align*}
\varphi_+(\xi)=\gamma_{-,\xi}\sigma_1\overline{\gamma_{+,\xi}\sigma_1\overline{\varphi_+(\xi)}}=\gamma_{-,\xi}\overline{\gamma_{+,\xi}}\varphi_{+}(\xi).
\end{align*}
The analyticity is a direct consequence of \eqref{23.4}.
\end{proof}

Let $P_\xi=\frac{1}{\sqrt{2}}\(\varphi_+(\xi)\ \varphi_-(\xi)\)$.
Then, if $P_\xi$ is invertible, it diagonalizes $T_{\lambda(\xi)}$.
\begin{align}\label{24}
P_\xi^{-1}T_{\lambda(\xi)} P_\xi=\begin{pmatrix} e^{\im \xi} & 0 \\ 0 & e^{-\im \xi} \end{pmatrix}.
\end{align}
We see that, $e^{\pm\im \xi x}\varphi_\pm(\xi)$ satisfies $e^{\pm\im \xi x}\varphi_\pm(\xi)=T_{\lambda(\xi)}  e^{\pm\im \xi (x-1)}\varphi_\pm(\xi)$.

Since we need later, we explicitly give $P_\xi$ and $P_\xi^{-1}$ for $\xi \in \T$.
\begin{align}
P_\xi=\frac{1}{\sqrt{2}}\begin{pmatrix} \overline{\alpha_0}/A_+ & \overline{\alpha_0}/A_-\\ (e^{\im \lambda(\xi)}-\rho_0 e^{\im \xi})/A_+ & (e^{\im \lambda(\xi)}-\rho_0 e^{-\im \xi})/A_-  \end{pmatrix},\label{24.1}\\
P_\xi^{-1}=\frac{1}{\im \sqrt{2}\overline{\alpha_0} \rho_0 \sin \xi} \begin{pmatrix}(e^{\im \lambda(\xi)}-\rho_0 e^{-\im \xi})/A_+ & -\overline{\alpha_0}/A_+ \\ (-e^{\im \lambda(\xi)}+\rho_0 e^{\im \xi})/A_- & \overline{\alpha_0}/A_- \end{pmatrix}.\label{24.2}
\end{align}

We set $R_0(\lambda)(x,y)$ to be the integral kernel of $R_0(\lambda)$, where
\begin{align*}
R_{0}(\lambda):=(\mathcal C_0-e^{\im \lambda})^{-1}.
\end{align*}
\begin{lemma}\label{lem:7}
Let $\xi\in \T_{\delta}$.
Then we have
\begin{align}\label{25}
R_0(\lambda(\xi))(x,y)=e^{\im \xi |x-y|}e^{-\im \lambda(\xi)}W_{0,\xi}^{-1}\(\sigma_1 Q_\xi^{\top} \sigma_1\begin{pmatrix} 1_{\leq y}(x) & 0 \\ 0 & 1_{<y}(x)\end{pmatrix}+Q_\xi \begin{pmatrix} 1_{> y}(x) & 0 \\ 0 & 1_{\geq y}(x)\end{pmatrix}\),
\end{align}
where
\begin{align*}
Q_\xi :=\varphi_+(\xi)\varphi_-(\xi)^\top\sigma_1,\quad W_{0,\xi} :=\det P_\xi=\det \(\varphi_+(\xi)\ \varphi_-(\xi)\).
\end{align*}
\end{lemma}

\begin{proof}
The formula \eqref{25} is direct consequence of Proposition \ref{prop:kernel}.
\end{proof}

\begin{remark}
We have
\begin{align}\label{25.1}
W_{0,\xi}=\frac{\im \overline{\alpha_0} \rho_0}{\sqrt{\(1-\rho_0\cos(\lambda(\xi)+\xi)\)\(1-\rho_0\cos(\lambda(\xi)-\xi)\)}}\sin \xi.
\end{align}
In particular, $W_{0,\xi}^{-1}\sin\xi \in C^\omega(\T;\R)$.
\end{remark}

\begin{proposition}\label{prop:8}
Let $\sigma>1/2$.
Then, $\<x\>^{-\sigma}R_0(\lambda(\xi))(x,y)\<y\>^{-\sigma} \in L^2_{x,y}(\Z^2;\mathcal L(\C^2))$ for $0\leq\Im \xi\leq \delta$, where $\mathcal L(\C^2)$ is the set of $2\times 2$ matrices.
Further, 
\begin{align*}
\T_{\delta,\geq 0}\setminus\{0,\pi\} \ni \xi \mapsto \<x\>^{-\sigma}R_0(\lambda(\xi))(x,y)\<y\>^{-\sigma} \in L^2_{x,y}(\Z^2;\mathcal L(\C^2))
\end{align*}
is $L^2_{x,y}(\Z^2;\mathcal L(\C^2))$ valued continuous function in $\T_{\delta,\geq 0}\setminus\{0,\pi\}$, where $\T_{\delta,\geq 0}:=\{z\in \T_{\delta}\ |\ \Im z \geq 0\}$. 
\end{proposition}

\begin{proof}
By \eqref{25}, for fixed $x,y\in \Z$, $\<x\>^{-\sigma}R_0(\lambda(\xi))(x,y)\<y\>^{-\sigma}$ is continuous with respect to $\xi$ and $$\|\<x\>^{-\sigma}R_0(\lambda(\xi))(x,y)\<y\>^{-\sigma}\|_{\mathcal L(\C^2)}\lesssim_\varepsilon \<x\>^{-\sigma}\<y\>^{-\sigma},$$ for $\xi\in T_{\delta,\geq 0}\setminus (B_{\C_p}(0,\varepsilon)\cup B_{\C_p}(\pi, \varepsilon))$ where the implicit constant is independent of $x,y$.
Therefore, by Lebesgue dominated convergence theorem we have the conclusion.
\end{proof}

Proposition \ref{prop:8} is the limiting absorption principle, for the unperturbed case.

Let $P_+$ to be the spectral projection to the upper arc of $\sigma(U_0)$ and $f$ analytic in the neighborhood of $\sigma(U_0)\cap\{z\in \C\ |\ \Im z>0\}$.
Then, we have
\begin{align*}
f(U_0)P_+=-\frac{1}{2\pi \im} \int_{\Gamma_\varepsilon} f(w)(U_0-w)^{-1}\,dw,
\end{align*}
where $\Gamma_\varepsilon:=\{e^{\im \lambda(\xi)}\ |\ \xi \in \T +\im\varepsilon\}$ encircles $\sigma(U_0)\cap\{z\in \C\ |\ \Im z>0\}$ in a clockwise direction. 
For the functional calculus, see Part I, Chapter I of \cite{GohbergBookI}.

By change of variables we have
\begin{align*}
f(U_0)P_+&=-\frac{1}{2\pi} J_{\mathrm{EV}}\int_{\tilde \Gamma_\varepsilon} f(e^{\im \lambda})R_0(\lambda) e^{\im \lambda}\,d\lambda J_{\mathrm{VE}}\\&
=-\frac{1}{2\pi}J_{\mathrm{EV}}\int_{\xi\in \T +\im\varepsilon}f(e^{\im \lambda(\xi)})e^{\im \lambda(\xi)}R_0(\lambda(\xi))\lambda'(\xi)\,d\xi J_{\mathrm{VE}},
\end{align*}
where $\tilde \Gamma_\varepsilon:=\{\lambda(\xi) \ |\ \xi\in \T+\im \varepsilon\}$.
Therefore, taking the limit $\varepsilon\downarrow 0$, the kernel of $J_{\mathrm{VE}} f(U_0)P_+ J_{\mathrm{EV}}$ becomes
\begin{align}\label{25.2}
-\frac{1}{2\pi}\int_{\T }f(e^{\im \lambda(\xi)})e^{\im \xi|x-y|} g_{x,y}(\xi)\,d\xi,
\end{align}
where
\begin{align*}
g_{x,y}(\xi)=\frac{\lambda'(\xi)}{ W_{0,\xi}}\(\sigma_1{}^tQ_\xi \sigma_1\begin{pmatrix} 1_{\leq y}(x) & 0 \\ 0 & 1_{<y}(x)\end{pmatrix}+Q_\xi \begin{pmatrix} 1_{> y}(x) & 0 \\ 0 & 1_{\geq y}(x)\end{pmatrix}\).
\end{align*}

The following dispersive estimate is a direct consequence of van der Corput lemma (see section 8.1 of \cite{SteinHarmonic}).
This is also proved in \cite{ST12JFA} and \cite{MSSSS18DCDS} in the context of QWs.
\begin{proposition}\label{prop:9}
We have $\|U_0^t P_+\|_{\mathcal L(l^1,l^\infty)}\lesssim \<t\>^{-1/3}$.
\end{proposition}

\begin{proof}
First, recall that we have $\lambda(\xi)=\arccos\(\rho_0\cos \xi\)$ and
\begin{align}\label{25.3}
\lambda'(\xi) = \frac{|a|\sin \xi}{\sqrt{1-|a|^2\cos^2\xi}},
\end{align}
and $\min(|\lambda''(\xi)|, |\lambda'''(\xi)|)\gtrsim 1$ (see \cite{MSSSS18DCDS}).
Substituting $f(s)=s^t$ into \eqref{25.2}, we have
\begin{align*}
J_{\mathrm{VE}} f(U_0)P_+ J_{\mathrm{EV}}(x,y)=-\frac{1}{2\pi}\int_\T e^{\im t \(\lambda(\xi)+\frac{|x-y|}{t}\xi\)}g_{x,y}(\xi)\,d\xi.
\end{align*}

Now, by \eqref{25.1}, $g_{x,y}(\xi)$ is analytic w.r.t.\ $\xi$ (and in particular $C^1$) and uniformly (in $x,y$) bounded in $C^1$.
Therefore, dividing the integral in the region where $|\lambda''(\xi)|\gtrsim 1$ and $|\lambda'''(\xi)|\gtrsim 1$ and applying the van der Corput lemma, we have the conclusion.
\end{proof}

%
%

\section{Jost solutions}\label{sec:4}
In this section we seek for solutions of 
\begin{align}\label{26}
	\phi_\pm(x+1,\xi)=T_{\lambda(\xi)}(x)\phi_\pm(x,\xi)
,\quad \phi_\pm(x,\xi)-e^{\pm\im \xi x}\varphi_\pm(\xi)\to 0,\text{ as }x\to \pm \infty,
\end{align}
where $\varphi_\pm(\xi)$ are given in \eqref{23.2}.
We set $\phi_\pm(x,\xi)=e^{\pm\im \xi x}m_\pm(x,\xi)$.
Then, \eqref{26} can be rewritten as
\begin{align}\label{27}
m_{\pm}(x,\xi)=A_{\xi,\pm}m_{\pm}(x-1,\xi)+V_{\xi,\pm}(x-1)m_{\pm}(x-1,\xi),
\end{align}
where
$A_{\xi,\pm}=e^{\mp \im \xi}T_{0,\lambda(\xi)} $ and $V_{\xi,\pm}(x-1)=e^{\mp \im \xi}\(T_{\lambda(\xi)}(x)-T_{0,\lambda(\xi)} \)$.
Notice that $A_{\xi,\pm}$ have eigenvalues $1$ and $e^{\mp 2\im \xi}$.

\begin{lemma}\label{lem:9.1}
Let $\xi\in \T_{\delta_0,\geq 0}:=\{z\in \C/2\pi\Z\ |\ 0\leq \Im z <\delta_0\}$.
Then, for $\mp x\geq 0$, we have
\begin{align*}
\|A_{\xi,\pm}^x\|_{\mathcal L(\C^2)}\lesssim \min( |x|, |\sin \xi|^{-1}).
\end{align*}
Here, the implicit constant is independent of $\xi \in \T_{\delta_0,\geq 0}$.
\end{lemma}

\begin{proof}
First, if $\xi \neq 0,\pi$, we can diagonalize $A_{\xi,\pm}$ as in \eqref{24} with $P_\xi$ given by \eqref{24.1}.
Thus,
\begin{align*}
\|A_{\xi,\pm}^x \|_{\mathcal L(\C^2)}\leq \|P_\xi\|_{\mathcal L(\C^2)} \|e^{\mp \im x \xi}\begin{pmatrix} e^{\im x \xi} & 0 \\ 0 & e^{-\im x \xi} \end{pmatrix}\|_{\mathcal L(\C^2)} \|P_\xi^{-1}\|_{\mathcal L(\C^2)}\lesssim |\sin \xi|^{-1}.
\end{align*}
The factor $|\sin \xi|^{-1}$ comes from $P_\xi$, see \eqref{24.2}.

Next, we set $\tilde P_\xi:=(\varphi_+(\xi)\ \tilde \varphi(\xi))$, where $\tilde \varphi(\xi)=A_+(\xi)^{-1}\begin{pmatrix} 0 \\ -\rho_0 \end{pmatrix}$ is the solution of 
\begin{align*}
(T_{0,\lambda(\xi)}-e^{-\im \xi})\tilde \varphi(\xi)=\varphi_+(\xi).
\end{align*}
Recall $\varphi_+(\xi)$ is defined in \eqref{23.2} and $A_+(\xi)$ is defined in \eqref{23.1.1}.
Indeed,
\begin{align*}
(T_{0,\lambda(\xi)}-e^{-\im \xi})\tilde \varphi(\xi)&=
A_+(\xi)^{-1}\rho_0^{-1}\begin{pmatrix} e^{\im \lambda(\xi)}-\rho_0e^{-\im \xi} & -\overline{\alpha_0} \\ -\alpha_0 & e^{-\im \lambda(\xi)}-\rho_0 e^{-\im \xi}\end{pmatrix}\begin{pmatrix} 0 \\ -\rho_0\end{pmatrix}\\&=A_+(\xi)^{-1}\begin{pmatrix} \overline{\alpha_0} \\ -e^{-\im \lambda(\xi)}+\rho_0e^{-\im \xi}\end{pmatrix}=A_+(\xi)^{-1}\begin{pmatrix} \overline{\alpha_0} \\ e^{\im \lambda(\xi)}-\rho_0 e^{\im \xi}\end{pmatrix}=\varphi_+(\xi),
\end{align*}
where the 2nd identity in the 2nd line follows from \eqref{21} (or \eqref{21.6}).
Therefore, we have
\begin{align*}
\tilde P_\xi&=A_+(\xi)^{-1}\begin{pmatrix} \overline{\alpha_0} & 0 \\ e^{\im \lambda(\xi)}-\rho_0 e^{\im \xi} & -\rho_0\end{pmatrix},
\\
\tilde P_\xi^{-1}&=\frac{A_+(\xi)}{\overline{\alpha_0} \rho_0}\begin{pmatrix} \rho_0 & 0 \\ e^{\im \lambda(\xi)}-\rho_0 e^{\im \xi} & -\overline{\alpha_0} \end{pmatrix}~,
\end{align*}
and
\begin{align*}
\tilde P_\xi^{-1}T_{0,\lambda(\xi)}\tilde P_\xi=\begin{pmatrix} e^{\im \xi} & 1 \\ 0 & e^{-\im \xi} \end{pmatrix}.
\end{align*}
Since
\begin{align}\label{27.6}
\begin{pmatrix} e^{\im \xi} & 1 \\ 0 & e^{-\im \xi} \end{pmatrix}^x=
\begin{cases}
\begin{pmatrix} e^{\im \xi x} & \sum_{y=0}^{x-1}e^{\im \xi (x-1-2y)} \\ 0 & e^{-\im \xi x} \end{pmatrix} & x> 0,\\
\begin{pmatrix} e^{\im \xi x} & -\sum_{y=0}^{-x-1}e^{-\im \xi (x+1+2y)} \\ 0 & e^{-\im \xi x} \end{pmatrix} & x< 0,
\end{cases}
\end{align}
we have
\begin{align*}
\|A_{\xi,\pm}^x \|_{\mathcal L(\C^2)}\leq \|\tilde P_\xi\|_{\mathcal L(\C^2)} \|e^{\mp \im\xi x }\begin{pmatrix} e^{\im \xi x} & \sum_{y=0}^{x-1}e^{\im(x-1-2y)} \\ 0 & e^{-\im \xi x} \end{pmatrix}\|_{\mathcal L(\C^2)} \|\tilde P_\xi^{-1}\|_{\mathcal L(\C^2)}\lesssim |x|.
\end{align*}
Recall that $A_+(\xi)\neq 0$ for $\xi \in \T$.
\end{proof}

\begin{lemma}\label{lem:10}
Let $\alpha(\cdot)-\alpha,\ \theta \in l^{1,\sigma}$.
Then, $\|V_{\xi,\pm}\|_{l^{1,\sigma}}\lesssim \|\alpha(\cdot)-\alpha\|_{l^{1,\sigma}}+\|\theta\|_{l^{1,\sigma}}$.
In particular, the implicit constant is independent of $\xi \in \T_{\delta_0,\geq 0}$.
\end{lemma}

\begin{proof}
By the definition of $V_{\xi,\pm}$, we have
\begin{align*}
\|V_{\xi,\pm}(x)\|_{\mathcal L(\C^2)}&\leq e^{|\Im \xi|} \rho_0^{-1}|\rho(x)-\rho_0| \|T_{\lambda(\xi)} (x)\|_{\mathcal L(\C^2)}\\&\quad +e^{|\Im \xi|} \rho_0^{-1}\left\|\begin{pmatrix} e^{\im \lambda(\xi)}(e^{-\im\theta(x)}-1) & -\overline{(\alpha(x)-\alpha)}\\ - (\alpha(x)-\alpha) & e^{-\im \lambda(\xi)}(e^{\im \theta(x)}-1) \end{pmatrix}\right\|_{\mathcal L(\C^2)}\\&\lesssim |\alpha(x)-\alpha|+|\theta(x)|.
\end{align*}
Notice that $|\Im \xi|$ is bounded for $\xi \in \T_{\delta_0,\geq 0} $.
Therefore, we have the conclusion.
\end{proof}

We use the discrete Duhamel formula.
That is, $u(t)=Au(t-1)+f(t-1)$ is equivalent to $u(t)=A^{t-t_0}u(t_0)+\sum_{s=t_0}^{t-1}A^{t-s-1}f(s)$.
\begin{lemma}\label{lem:11}
Let $\xi \in \T_{\delta_0,\geq 0} $.
Then, $m_{+}(\cdot,\xi)$ satisfies \eqref{27} with $m_{+}(x,\xi)\to \varphi_+(\xi)$ as $x\to \infty$ if and only if
\begin{align}\label{28}
m_{+}(x,\xi)=\varphi_{+}(\xi)-\sum_{y=x}^\infty A_{\xi ,+}^{-(y+1-x)} V_{\xi ,+}(y)m_{+}(y,\xi).
\end{align}
Similarly, $m_-(x,\xi)$ satisfies \eqref{27} with $m_-(x,\xi)\to \varphi_-(\xi)$ as $x\to -\infty$ if and only if
\begin{align}\label{29}
m_-(x,\xi)=\varphi_-(\xi)+\sum_{y=-\infty}^{x-1}A_{\xi ,-}^{x-y-1}V_{\xi ,-}(y)m_-(y,\xi).
\end{align}
\end{lemma}

\begin{proof}
It is obvious that if $m_+$ (resp.\ $m_-$) satisfies \eqref{28} (\eqref{29}), then $m_+$ ($m_-$) satisfies \eqref{27} with $m_{+}(x,\xi)\to \varphi_+(\xi)$ as $x\to \infty$ ($m_{-}(x,\xi)\to \varphi_-(\xi)$ as $x\to - \infty$).

Next, let $m_\pm$ satisfying \eqref{27} and $m_{\pm}(x,\xi)\to \varphi_\pm(\xi)$ as $x\to \pm\infty$.
Set $\psi_{\xi,\pm}(x)=m_\pm(x,\xi)-\varphi_\pm(\xi)$.
Then, for $x_1>x_2$, we have
\begin{align}\label{30}
\psi_{\xi ,\pm}(x_1)=A_{\xi ,\pm}^{x_1-x_2}\psi_{\xi ,\pm}(x_2)+\sum_{y=x_2}^{x_1-1}A_{\xi ,\pm}^{x_1-y-1}V_{\xi ,\pm}(y)\(\varphi_{\xi ,\pm}+\psi_{\xi,\pm}(y)\).
\end{align}
Notice that \eqref{30} is equivalent to \eqref{27}.
Setting $x_1=x$, $x_2\to -\infty$ and adding $\varphi_-(\xi)$ to both sides, we have \eqref{29}.
To obtain \eqref{28}, we multiply $A_{\xi ,+}^{-(x_1-x_2)}$ to \eqref{30} and set $x_2=x$, $x_1\to \infty$ and finally add $\varphi_+(\xi)$ to both sides.
\end{proof}

\begin{proposition}\label{prop:12}
If $\alpha(\cdot)-\alpha, \theta(\cdot)\in l^1$, there exist solutions of \eqref{28} and  \eqref{29} for $\xi\in \T_{\delta_0,\geq 0}\setminus\{0,\pi\}$ and we have
\begin{align}\label{31.0}
\|m_\pm(x,\xi)\|_{\C^2}\leq C(|\sin \xi|^{-1}),
\end{align}
where $C(s)$ is an increasing function of $s$.
If $\alpha(\cdot)-\alpha, \theta(\cdot)_0\in l^{1,1}$, there exist solutions of \eqref{28} and \eqref{29} for $\xi\in \T_{\delta_0,\geq 0}$.
Further, we have
\begin{align}\label{32}
\|m_\pm (x,\xi)\|_{\C^2}\lesssim \max(1,\mp x),
\end{align}
where the implicit constant is independent of $\xi\in \T_{\delta_0,\geq 0}$.
\end{proposition}

\begin{proof}
We will only prove the Proposition for $m_+$.
First, we assume $\alpha(\cdot)-\alpha_0, \theta(\cdot)\in l^1$.
Then, from Lemma \ref{lem:10} we have $V_{\xi,+} \in l^1$.
Now, take $x_{\pm,\xi} \in \N$ so that we have 
\begin{align}\label{32.1}
\sum_{ \pm x\geq x_{\pm,\xi}} \|V_{\xi,+}(x)\|_{\mathcal L(\C^2)}\ll |\sin \xi|.
\end{align}
Set
\begin{align*}
\Phi_{\xi,+}(m)(x):=\varphi_{+}(\xi)-\sum_{y=x}^\infty A_{\xi ,+}^{-(y+1-x)} V_{\xi ,+}(y)m(y).
\end{align*}
We show $\Phi_{\xi,+}$ is a contraction mapping in $B_1:=B_{l^\infty(\Z_{\geq x_{+,\xi}};\C^2)}\(2\|\varphi_+(\xi)\|_{\C^2}\)$.
Indeed, we have
\begin{align*}
\|\Phi_{\xi,+}(0)\|_{l^\infty(\Z_{\geq x_{+,\xi}};\C^2)}=\|\varphi_+(\xi)\|_{\C^2}.
\end{align*}
Further, by Lemmas \ref{lem:9.1} and \eqref{32.1}, we have
\begin{align*}
&\|\Phi_{\xi,+}(m_1)-\Phi_{\xi,+}(m_2)\|_{l^\infty(\Z_{\geq x_{+,\xi}};\C^2)}\\&\leq \sup_{x\geq x_{+,\xi}}\sum_{y=x}^\infty \|A_{\xi ,+}^{-(y+1-x)} \|_{\mathcal L(\C^2)} \|V_{\xi ,+}(y)\|_{\mathcal L(\C^2)} \|m_1-m_2\|_{l^\infty(\Z_{\geq x_{+,\xi}};\C^2)}\nonumber\\&
\leq   \frac 1 2  \|m_1-m_2\|_{l^\infty(\Z_{\geq x_{+,\xi}};\C^2)}.
\nonumber
\end{align*}
Thus, we see that $\Phi_{\xi,+}$ is a contraction mapping in $B_1$.
Next, we extend $m_+(x,\xi)$ defined for $x\geq x_{+,\xi}$ to $x\geq -x_{-,\xi}$ by \eqref{27}.
Then, since $m_+(x_{+,\xi},\xi)=e^{-\im \xi(x_{+,\xi}-x)}T_{\lambda(\xi)}(x_{+,\xi})\cdots T_{\lambda(\xi)}(x+1)m_+(x,\xi)$, by Lemma \ref{lem:9.1}, we have a crude bound
\begin{align}
\sup_{x\geq -x_{-,\xi}} \|m_+(x,\xi)\|_{\C^2}&\leq 2 \(\max_{-x_{-,\xi}\leq x\leq x_{+,\xi}}\|T_{\lambda(\xi)}(x)^{-1}\|_{\mathcal L(\C^2)}\)^{x_{+,\xi}+x_{-,\xi}}\|\varphi_+(\xi)\|_{\C^2}\nonumber\\
&\leq \tilde C(|\sin \xi|^{-1}),\label{35.0}
\end{align}
where $\tilde C(s)$ is a increasing function of $s$.
By \eqref{30}, for $x<-x_{-,\xi}$ we have
\begin{align}\label{35.1}
m_+(x,\xi)=\varphi_+(\xi)+A_{\xi,+}^{x_{-,\xi}+x}(m_+(x_{-,\xi},\xi)-\varphi_+(\xi))-\sum_{y=x}^{-x_{-,\xi} -1} A_{\xi,+}^{x-y-1} V_{\xi,+}(y) m_{+}(y,\xi).
\end{align}
By Lemma \ref{lem:9.1} and \eqref{35.0}, we have
\begin{align*}
\|m_+(x,\xi)\|_{\C^2}\leq &\|\varphi_+(\xi)\|_{\C^2}+C|\sin \xi|^{-1}(\tilde C(|\sin \xi|^{-1})+\|\varphi_+(\xi)\|_{\C^2})\\&+C\sup_{x\leq y<x_{-,\xi}} \|m_{\xi}(y,\xi)\|_{\C^2}\sum_{y=x}^{-x_{-,\xi} -1} |\sin \xi|^{-1} \|V_{\xi,+}(y)\|_{\mathcal L(\C^2)} .
\end{align*}
By \eqref{32.1}, we have we have $C\sum_{y=x}^{-x_{-,\xi} -1} |\sin \xi|^{-1} \|V_{\xi,+}(y)\|_{\mathcal L(\C^2)} \leq \frac 1 2 $.
Therefore, we have for any $x_0<x_{-,\xi}$
\begin{align*}
\sup_{x_0\leq x<x_{-,\xi}}\|m_+(x,\xi)\|_{\C^2}\leq & C\(|\sin \xi|^{-1}+1\)\|\varphi_+(\xi)\|_{\C^2}+C|\sin \xi|^{-1}\tilde C(\sin \xi) \\&+ \frac 1 2 \sup_{x_0\leq x<x_{-,\xi}}\|m_+(x,\xi)\|_{\C^2}.
\end{align*}
Therefore, combined with  \eqref{35.0}, we have \eqref{31.0}.

We next assume $\alpha(\cdot)-\alpha_0\in l^{1,1}$.
Then we have $V_{\xi,+}\in l^{1,1}$.
Take $x_\pm>0$ (independent of $\xi$) so that $\sum_{\pm x\geq x_\pm} |x|\|V_{\xi,+}(x)\|_{\mathcal L(\C^2)}\ll1$.
Then, we can show $\Phi_{\xi,+}$ is a contraction mapping in $B_2=B_{l^\infty(\Z_{\geq x_+};\C^2)}\(2\|\varphi_{+,\xi}\|_{\C^2}\)$.
Indeed,
\begin{align*}
&\|\Phi_{\xi,+}(m_1)-\Phi_{\xi,+}(m_2)\|_{l^\infty(\Z_{\geq x_{+}};\C^2)}\\&\lesssim \sup_{x\geq x_{+,\xi}}\sum_{y=x}^\infty \<y-x+1\> \|V_{\xi ,+}(y)\|_{\mathcal L(\C^2)} \|m_1-m_2\|_{l^\infty(\Z_{\geq x_{+}};\C^2)}\nonumber\\&
\lesssim \|V\|_{l^{1,1}_{x_+}}\sup_{x\geq x_{+,\xi}}\|m_1-m_2\|_{l^\infty(\Z_{\geq x_{+}};\C^2)}
\leq   \frac 1 2  \|m_1-m_2\|_{l^\infty(\Z_{\geq x_{+,\xi}};\C^2)}.
\end{align*}
Next, we extend $m_+(x,\xi)$ to $x=-x_-$ by \eqref{27}.
Then, by \eqref{35.0} we have $\sup_{x\geq x_-}\|m_+(\cdot,\xi)\|_{l^\infty}\leq C$ where $C$ is independent of $\xi$ because $x_\pm$ are now independent of $\xi$ and $\|T_{\lambda(\xi)}(x)^{-1}\|_{\mathcal L(\C^2)}$ is uniformly bounded.
For $x\leq -x_-$ we have \eqref{35.1} with $x_{-,\xi}$ replaced by $x_-$.
Thus, dividing by $\<x\>$, we have
\begin{align*}
\sup_{x\geq x_0}\frac{\|m_+(x,\xi)\|_{\C^2}}{\<x\>}\lesssim 1 + \|m_+(-x_-,\xi)\|_{\C^2}+\frac 1 2 \sup_{x\geq x_0}\frac{\|m_+(x,\xi)\|_{\C^2}}{\<x\>}.
\end{align*}
Therefore, we obtain the conclusion.
\end{proof}

We now introduce the Wiener algebra $\mathcal A$.
For $u \in L^1(\T)$, we set
\begin{align*}
\|u\|_{\mathcal A}:=\sum_{n\in \Z}| \hat u(n)|,\quad \hat u(n):=\frac{1}{2\pi}\int_\T  u(\xi)e^{-\im n \xi}\,d\xi,
\end{align*}
and $\mathcal A:=\{ u\in L^1(\T)\ |\ \|u\|_{\mathcal A}<\infty\}$.
We further set $\mathcal A^1\subset \mathcal A$ by the norm $\|u\|_{\mathcal A^1}:=\|\hat u\|_{l^{1,1}}$.

\begin{lemma}\label{lem:18}
Let $\alpha(\cdot)-\alpha_0, \theta(\cdot)\in l^{1,\sigma}$, $\sigma=1,2$.
Then, for $0\leq \delta<\delta_0$ and $\mp x\geq 0$, we have
\begin{align}\label{WienerEst1}
\|A_{\cdot+\im \delta,\pm}^x\|_{\mathcal A^{\sigma-1}}\lesssim \<x\>^\sigma,\quad \sum_{y\in \Z}\<x\>^\sigma\|V_{\cdot+\im \delta,\pm}(x)\|_{\mathcal A^{\sigma-1}}<\infty,
\end{align}
where we have set $\mathcal A^0:=\mathcal A$.
\end{lemma}

\begin{proof}
The proof follows immediately from \eqref{27.6} for $A_{\cdot,\pm}$.
Also we can bound $\|V_{\cdot,\pm}\|_{\mathcal A^{\sigma-1}}$ by the same manner as Lemma \ref{lem:10}.
\end{proof}

\begin{proposition}\label{prop:19}
Let $\alpha(\cdot)-\alpha_0, \theta(\cdot)\in l^{1,\sigma}$, $\sigma=1,2$.
Then, for $0\leq \delta<\delta_0$, we have
\begin{align}\label{Abound}
\|m_\pm(x,\cdot+\im \delta)\|_{\mathcal A^{\sigma-1}}\lesssim \max(1,\mp x|x|^{\sigma-1}).
\end{align}
Moreover, for each $x\in \Z$, we have $m_\pm(x,\cdot+\im \delta)\to m_\pm(x,\cdot)$ in $\mathcal A^{\sigma-1}$.

\end{proposition}

\begin{proof}
The proof of \eqref{Abound} is similar to the proof of \eqref{32} of Proposition \ref{prop:12}.
Indeed, we only have to replace the bound given in Lemmas \ref{lem:9.1} and \ref{lem:10} by the bound \eqref{WienerEst1} given by Lemma \ref{lem:18}.
The second claim follows from the continuity of $\varphi(\cdot+\im \delta)$, $A_{\cdot+\im \delta,\pm}$ and $V_{\cdot+\im \delta,\pm}(x)$ in $\mathcal A^{\sigma-1}$.
\end{proof}
%
%

\section{Scattering theory}\label{sec:scat}

By the Jost solutions, we can perform the scattering theory for QWs.
We start from considering the zero of 
\begin{align*}
W_\xi:=\det(\phi_+(\cdot,\xi),\phi_-(\cdot,\xi)).
\end{align*}
Recall that by Proposition \ref{prop:Wronsky}, $W_\xi$ is independent of $x$.

\begin{lemma}\label{lem:14}
Let $\alpha(\cdot)-\alpha_0,\ \theta(\cdot)\in l^1$.
Let $\xi\in \T\setminus\{0,\pi\}$.
Then, we have $W_\xi\neq 0$.
\end{lemma}

\begin{proof}
For $\xi\in \T\setminus\{0,\pi\}$, recall that we have $\lambda(-\xi)=\lambda(\xi)$.
Thus, $\phi_\pm(\cdot,-\xi)$ satisfies \eqref{eq:equiv}. 
Thereore, there exists $r_\pm(\xi)$ and $t_\pm(\xi)$ s.t.
\begin{align}
\phi_+(\cdot,-\xi)&=-r_+(\xi)\phi_+(\cdot,\xi)+t_+(\xi)\phi_-(\cdot,\xi),\label{40}\\
\phi_-(\cdot,-\xi)&=t_-(\xi)\phi_+(\cdot,\xi)-r_-(\xi)\phi_-(\cdot, \xi).\label{41}
\end{align}
Taking $\det(\phi_+(\cdot,-\xi)\ \phi_+(\cdot,\xi))$ and $\det(\phi_-(\cdot,-\xi)\ \phi_-(\cdot,\xi))$, we have
\begin{align*}
W_{0,\xi}=t_\pm(\xi) W_\xi.
\end{align*}
Since $W_{0,\xi}\neq 0$ for $\xi\in \T\setminus\{0,\pi\}$, we have $W_\xi\neq 0$ and $t(\xi):=\frac{W_{0,\xi}}{W_\xi}=t_+(\xi)=t_-(\xi)$.
\end{proof}

\begin{definition}
We call
\begin{align}\label{42.2}
t(\xi)=\frac{W_{0,\xi}}{W_\xi},\quad r_\pm(\xi):=-\frac{\det(\phi_+(\cdot,\mp \xi), \phi_-(\cdot,\pm \xi))}{W_\xi}
\end{align}
the transmission and reflection coefficients respectively.
\end{definition}

\begin{lemma}\label{lem:15}
Let $\alpha(\cdot)-\alpha_0,\ \theta(\cdot)\in l^1$ and $\xi \in \T\setminus\{0,\pi\}$.
Then, we have $|t(\xi)|^2+|r_\pm(\xi)|^2=1$.
\end{lemma}

\begin{proof}
First, since $\xi \in \T\setminus\{0,\pi\}$ we have $\lambda(\xi)\in \R$.
By the symmetry \eqref{16} we see that $\sigma_1\overline{\phi_\pm(\cdot,\xi)}$ satisfies $\phi(x)=T_{\lambda(\xi)}(x)\phi(x-1)$.
Further by Lemma \ref{lem:6.3}, we have 
\begin{align*}
\sigma_1\overline{\phi_{\pm}(x,\xi)}-\gamma(\xi)^{-1}e^{\pm \im x (-\xi)}\varphi_\pm(-\xi)=\sigma_1\overline{\phi_{\pm}(x,\xi)}-\sigma_1\overline{e^{\pm \im x\xi}\varphi_\pm(\xi)}\to 0,\quad x\to \pm \infty.
\end{align*}
Therefore, we have
\begin{align}\label{43}
\gamma(\xi)\sigma_1\overline{\phi_\pm(x,\xi)}= \phi_\pm(x,-\xi).
\end{align}
Now, by \eqref{40}, \eqref{41} and \eqref{43}, we have
\begin{align*}
&\phi_\pm(\cdot,\xi)=\frac{1}{t(\xi)}\( \gamma_\xi\sigma_1\overline{\phi_\mp(\cdot,\xi)}+r_\mp(\xi)\phi_\mp(\cdot,\xi)\)\\&=
\frac{1}{t(\xi)}\( \gamma_\xi\sigma_1\overline{\frac{1}{t(\xi)}\( \gamma_\xi\sigma_1\overline{\phi_\pm(\cdot,\xi)}+r_\pm(\xi)\phi_\pm(\cdot,\xi)\)}+r_\mp(\xi)\frac{1}{t(\xi)}\( \gamma_\xi\sigma_1\overline{\phi_\pm(\cdot,\xi)}+r_\pm(\xi)\phi_\pm(\cdot,\xi)\)\)\\&=
\(\frac{1}{|t(\xi)|^2}+\frac{r_\pm(\xi)r_\mp(\xi)}{t(\xi)^2}\)\phi_\pm(\cdot,\xi)+\(\frac{\bar r_\pm(\xi)}{|t(\xi)|^2}+\frac{r_\mp(\xi)}{t^2(\xi)}\)\gamma_\xi \sigma_1 \overline{\phi_\pm(\cdot,\xi)}.
\end{align*}
Therefore, we have
\begin{align*}
\frac{1}{|t(\xi)|^2}+\frac{r_\pm(\xi)r_\mp(\xi)}{t(\xi)^2}=1,\quad \frac{\bar r_\pm(\xi)}{|t(\xi)|^2}+\frac{r_\mp(\xi)}{t^2(\xi)}=0.
\end{align*}
This implies $|t(\xi)|^2+|r_\pm(\xi)|^2 =1$.
\end{proof}

The nonexistence of embedded eigenvalues can be proved by the existence of Jost solutions.
We set $E$ to be the edge of the intervals $I_1,I_2$ given in \eqref{21.2}.
\begin{proposition}[Proposition \ref{prop:spec} (ii)]\label{prop:16}
Let $\alpha(\cdot)-\alpha_0,\  \theta(\cdot)\in l^1$.
Then, for $\lambda\in I_0\setminus E$, $e^{\im \lambda}$ is not an eigenvalue.
Further, if $\alpha(\cdot)-\alpha_0,\  \theta(\cdot) \in l^{1,1}$, for $\lambda\in I_0$, $e^{\im \lambda}\in I_0$ is not an eigenvalue.
\end{proposition}

\begin{proof}
We only consider $\lambda\in I_1$.
By the definition of $I_1$ and $\lambda(\xi)=\lambda_+(\xi)$, we have $I_1=\{\lambda(\xi)\ |\ \xi\in \T$ and $I_1\cap E=\{\lambda(\xi)\ |\ \xi=0,\pi\}$.
Set $\xi$ be such that $\lambda(\xi)=\lambda$.
In both cases, by Proposition \ref{prop:12}, there exists $\phi_+(x,\xi)$ which is a solution of \eqref{eq:equiv} bounded in $x\geq 0$.
Therefore, the coclusion immediately follows from Corollary \eqref{cor:Wron}.
\end{proof}

By Proposition \ref{prop:kernel}, we have the expression of the kernel of $R(\lambda)=(\mathcal C-e^{\im \lambda})^{-1}$.
\begin{lemma}\label{lem:13.1}
Let $\xi\in \T_{\delta_0,\geq 0}\setminus\{0,\pi\}$.
Then we have
\begin{align*}
&R(\lambda(\xi))(x,y)
=e^{\im \xi|x-y|}e^{-\im \lambda(\xi)}W_{\xi}^{-1}\\&\quad\times\(m_-(x,\xi)m_+(y,\xi)^\top \sigma_1 \begin{pmatrix} 1_{\leq y}(x) & 0 \\ 0 & 1_{<y}(x)\end{pmatrix}+m_+(x,\xi)m_-(y,\xi)^\top  \sigma_1 \begin{pmatrix} 1_{> y}(x) & 0 \\ 0 & 1_{\geq y}(x)\end{pmatrix}\).\nonumber
\end{align*}
Moreover, for $\sigma>1/2$, $\<x\>^{-\sigma}R(\lambda(\xi))(x,y)\<y\>^{-\sigma} \in L^2_{x,y}(\Z^2;\mathcal L(\C^2))$.
Further, it is $L^2_{x,y}(\Z^2;\mathcal L(\C^2))$ valued holomorphic function in $\T_{\delta_0,>0}$ and continuous in $\T_{\delta_0,\geq 0}\setminus\{0,\pi\}$. 
\end{lemma}

\begin{proof}
We only prove the continuity.
Notice that for $\xi\in \T\setminus \{0,\pi\}$, $m_\pm(\cdot,\xi)$ are bounded from Proposition \ref{prop:12}.
Therefore, since we have the pointwise limit and a uniform bound for the $L^2_{x,y}$ norm, the continuity follows from Lebegue dominated convergence theorem.
\end{proof}

\begin{remark}
Limiting absorption principle implies the nonexistence of singular continuous spectrum (Proposition \ref{prop:spec} (i)).
\end{remark}

For $\lambda \in E$, we say $\lambda$ is a resonance iff there exists a bounded solution of $(\cC-e^{\im \lambda})\phi=0$ (Definition \ref{def:resonance}).
The following lemma gives an necessary and sufficient condition for the existence of resonance.
\begin{lemma}\label{lem:17}
Let $\alpha(\cdot)-\alpha_0,\ \theta(\cdot)\in l^{1,1}$.
Let $\xi=0,\pi$.
Then, $W_\xi=0$ if and only if there exists a nontrivial bounded solution of \eqref{eq:equiv}.
\end{lemma}

For the proof, we follow \cite{CT09SIAM}.

\begin{proof}
If $W_\xi=0$. Then $\phi_+(\cdot,\xi)$ and $\phi_-(\cdot,\xi)$ are linearly dependent.
By Proposition \ref{prop:12}$, \phi_\pm(\cdot,\xi)$ is bounded in $\pm x\geq 0$.
Thus, we see that $\phi_+(\cdot,\xi)$ is bounded.

Suppose that there exists nontrivial bounded solution $\phi$ of $(\cC-e^{\im \lambda(\xi)})\phi=0$.
For contradiction, we assume $W_\xi\neq 0$.
Then, we can express $\phi=c_+\phi_++c_-\phi_-$ by some constants $c_\pm$.
If $c_+=0$, we can take $\phi=\phi_-$, which means $\phi_\pm$ are both bounded in $\Z_{\geq 0}$.
If $c_+\neq 0$, $\phi$ and $\phi_-$ will be linear independent solutions bounded in $\Z_{\leq 0}$.
Without loss of generality, we can assume there exists two linear independent solutions $\phi_j$ ($j=1,2$) bounded in $\Z_{\geq 0}$.
Further, we can assume $\|\phi_j\|_{l^\infty(\Z_{\geq 0})}=1$.

Now, we set $K$ by Proposition \ref{prop:kernel} with $\varphi_j$ replaced by $\phi_j$.
For $x_0>0$ (chosen later) and given $\phi\in l^\infty(\Z_{\geq x_0})$, we set $\Phi_\phi:l^\infty(\Z_{\geq x_0})\to l^\infty(\Z_{\geq x_0})$ by
\begin{align}\label{53}
\Phi_\phi(\psi)(x):=\phi(x)+\sum_{y\geq x} K(x,y) \(\mathcal C -\mathcal C_0\)\psi (y) .
\end{align}
Here, we have extended $\psi$ to $\Z_{\geq x_0-1}$ by setting $\psi(x_0-1)=0$ (notice that because of Lemma \ref{lem:cmv} for $x=x_0$, we need $\psi(x_0-1)$ in the r.h.s.\ of \eqref{53}).

Now, by Lemma \ref{lem:cmv} and assumption $\alpha(\cdot)-\alpha_0,\ \theta(\cdot)\in l^{1,1}$, there exists $v\in l^{1,1}$ s.t.\ $$\|\(\cC-\cC_0\)\psi(y)\|_{\C^2}\leq v(y)\max_{a=0,\pm 1}\|\psi(y+a)\|_{\C^2}.$$
Therefore, we have
\begin{align*}
\|\Phi_\phi(\psi_1)-\Phi_\phi(\psi_2)\|_{l^\infty(\Z_{\geq x_0})}\leq \sup_{x,y\geq x_0}\|K(x,y)\|_{\mathcal L(\C^2)} \sum_{y\geq x_0}v(y)\|\psi_1-\psi_2\|_{l^\infty(\Z_{\geq x_0})}
\end{align*}
Since, $\sup_{x,y\geq x_0}\|K(x,y)\|_{\mathcal L(\C^2)}<\infty$, there exists $x_0$ sufficiently large (independent of $\phi$) s.t.
\begin{align*}
\|\Phi_\phi(\psi_1)-\Phi_\phi(\psi_2)\|_{l^\infty(\Z_{\geq x_0})}\leq \frac 1 2 \|\psi_1-\psi_2\|_{l^\infty(\Z_{\geq x_0})}.
\end{align*}
Therefore, we have a fixed point of $\Phi_\phi$.
For $\phi\in l^\infty(\Z_{\geq x_0})$, we denote the fixed point of $\Phi_\phi$ by $\psi(\phi)$.
Then, $\psi$ is a linear mapping and moreover because
\begin{align*}
\|\psi(\phi)\|_{l^\infty(\Z_{\geq x_0})}\geq \|\phi\|_{l^\infty(\Z_{\geq x_0})}-\frac 1 2 \|\psi(\phi)\|_{l^\infty(\Z_{\geq x_0})},
\end{align*}
we have
\begin{align*}
\|\psi(\phi)\|_{l^\infty(\Z_{\geq x_0})}\geq \frac 2 3\|\phi\|_{l^\infty(\Z_{\geq x_0})}
\end{align*}
and in particular, it is an injection.

Now, if we restrict $\psi$ to the $2$ dimensional space which is spanned by $\phi_j$ ($j=1,2$), we have
$\(\cC_0-e^{\im \lambda(\xi)}\)\psi(\phi)(x)=0$ ($x\geq x_0+1$) for $\phi=c_1\phi_1+c_2\phi_2$.
Indeed,
\begin{align*}
&\(\cC_0-e^{\im \lambda( \xi)}\)\psi(\phi)(x)=(\cC-e^{\im \lambda(\xi)})\phi(x)-(\cC-\cC_0)\phi(x)\\&\quad+\(\cC-e^{\im \lambda(\xi)}\)\sum_{y\geq x}K_\lambda(x,y)\((\cC-\cC_0)\psi(\phi)\)(y)
-\(\cC-\cC_0\)\sum_{y\geq x}K_\lambda(x,y)\((\cC-\cC_0)\psi(\phi)\)(y)\\&=
(\cC-\cC_0)\psi(\phi)(x)-\(\cC-\cC_0\)\(\phi(x)+\sum_{y\geq x} K(x,y) \(\mathcal C -\mathcal C_0\)\psi(\phi) (y) \)=0,
\end{align*}
where we have used $(\cC-e^{\im \lambda(\xi)})\phi(x)=0$ and the fact that $K$ is the kernel of $(\cC-e^{\im \lambda(\xi)})$ in the second equality.
Therefore, $\psi$ maps the solutions of $(\cC-e^{\im \lambda(\xi)})u=0$ to the solutions of $(\cC_0-e^{\im \lambda(\xi)})u=0$ (in the region $x>x_0$).
Further, $\psi$ is an injection so $(\cC_0-e^{\im \lambda(\xi)})u=0$ has to have $2$ dimensional solutions in $l^\infty(\Z_{x\geq x_0})$.
However, since it is $1$ dimensional, we have a contradiction.
\end{proof}

\begin{proposition}\label{prop:20}
Let $\alpha(\cdot)-\alpha_0,\ \theta(\cdot)\in l^{1,1}$ for the generic case and $\alpha(\cdot)-\alpha_0,\ \theta(\cdot)\in l^{1,2}$ for the nongeneric case. 
Then, we have
$t,r_\pm \in \mathcal A$.
\end{proposition}

\begin{proof}
We divide the proof for the cases $W_0 W_\pi \neq 0$ and $W_0W_\pi=0$.
The first case is simple.
Indeed, by Proposition \ref{prop:19}, we have $W_\xi \in \mathcal A$.
Further, from Lemma \ref{lem:14} and the assumption $W_0 W_\pi \neq 0$, we have $W_\xi\neq 0$ for all $\xi \in \T$.
Therefore, we have $W_\xi^{-1}\in \mathcal A$ (for the properties of Wiener algebra, see \cite{KatznelsonBook}).
By \eqref{42.2}, we have the conclusion in this case.

We now consider the nongeneric case
and assume $W_0=0$.
By Fourier expansion with respect to the $\xi$ variable, we have
\begin{align*}
\phi_\pm(x,\xi)=\sum_{n\in \Z} \hat \phi_\pm(x,n) e^{\im n \xi},\quad \hat \phi_\pm(x,n)=\frac{1}{2\pi}\int_\T \phi(x,\xi)e^{-\im n \xi}\,d\xi.
\end{align*}
By Proposition \ref{prop:19} we have $\{\hat \phi_\pm(x,n)\}_{n\in\Z}\in l^{1,1}$.
subtracting $\phi_\pm(x,0)$ we have
\begin{align*}
\phi_\pm(x,\xi)-\phi_\pm(x,0)&=\sum_{n\neq 0}\hat \phi_\pm(x,n)\(e^{\im n\xi}-1\)\\&=(e^{\im \xi}-1)\(\sum_{n>0}\hat \phi_\pm(x,n)\sum_{m=0}^{n-1}e^{\im m \xi}+\sum_{n<0}\hat \phi_\pm(x,n)\sum_{m=n+1}^{0}e^{\im m \xi}\)\\&
=(e^{\im \xi}-1)\(\sum_{m=0}^\infty \(\sum_{n>m}\hat \phi_\pm(x,n)\)e^{\im m \xi}+\sum_{m=-\infty}^0\(\sum_{n<m}\hat \phi_\pm(x,n)\)e^{\im m \xi}\)
\end{align*}
Now set
\begin{align*}
\psi_\pm(x,m)=\begin{cases} \sum_{n>m}\hat \phi_\pm(x,n) ,& m>0,\\
\sum_{n\neq 0} \hat \phi_\pm(x,n), & m=0,\\
\sum_{n< m} \hat \phi_\pm(x,n), & m<0.
\end{cases}
\end{align*}
Then, we have $\{\psi_\pm(x,m)\}_{m\in \Z}\in l^1$.
Indeed, by $\{\hat \phi_\pm(x,n)\}_{n\in\Z}\in l^{1,1}$,
\begin{align*}
\sum_{m>0}|\psi_\pm(x,m)|\leq \sum_{m>0}\sum_{n>m}|\hat \phi(x,n)|=\sum_{m>0}(m-1)|\hat \phi(x,m)|<\infty.
\end{align*}
Therefore, we have
\begin{align*}
\phi_\pm(x,\xi)-\phi_\pm(x,0)=(e^{\im \xi}-1)\Psi_\pm(x,\xi),
\end{align*}
where $\Psi_\pm(x,\cdot)\in \mathcal A$.

Now, since $W_0=0$, we have
\begin{align}
W_\xi&=\det(\phi_+(x,\xi)\ \phi_-(x,\xi))\nonumber\\&=(e^{\im \xi}-1)\(\det(\phi_+(x,\xi)\ \Psi_-(x,\xi))+\det(\Psi_+(x,\xi)\ \phi_-(x,0))\),\nonumber\\
&=(e^{\im \xi} -1)\tilde \Psi(\xi),\label{54}\\
\det(\phi_+(x,\mp\xi)\ \phi_-(x,\pm\xi))&=(e^{\im \xi}-1)\(\det(\phi_+(x,\mp\xi)\ \Psi_-(x,\pm\xi))+\det(\Psi_+(x,\mp\xi)\ \phi_-(x,0))\)\nonumber\\
&=(e^{\im \xi} -1)\tilde \Psi_\pm(\xi).\nonumber
\end{align}
Now, suppose that $\tilde \Psi(0)=0$.
Then, since
\begin{align*}
t(\xi)=\frac{W_{0,\xi}}{e^{\im \xi}-1} \frac{1}{\tilde \Psi(\xi)},
\end{align*}
and $|t(\xi)|\leq 1$, $|\frac{W_{0,\xi}}{e^{\im \xi}-1}|$ is bounded in $\xi\in \T$, we have a contradiction.
Therefore, 
\begin{align}\label{57}
\tilde \Psi(0)\neq 0.
\end{align}
Thus, let $\chi,  \tilde \chi, \tilde{\tilde \chi}\in C^\infty(\T)$ s.t.\ $\mathrm{supp}\chi\subset (-\pi/4,\pi/4)$, $\mathrm{supp}\tilde \chi\subset (-\pi/2,\pi/2)$, $\tilde \chi(\xi)=1$ in $\tilde\xi \in (-\pi/4,\pi/4)$ and $\mathrm{supp} \tilde{\tilde \chi}\subset \T\setminus (-\pi/4,\pi/4)$.
We choose $\tilde {\tilde }\chi$ so that $\tilde \chi(\xi)\tilde \Psi(\xi)+\tilde{\tilde \chi}(\xi)\neq 0$ for all $\xi \in \T$.
Then, we have $(\tilde \chi(\cdot)\tilde \Psi(\cdot)+\tilde{\tilde \chi}(\cdot))^{-1}\in \mathcal A$ and
\begin{align*}
\tilde \chi(\xi) t(\xi)=\frac{W_{0,\xi}}{e^{\im \xi}-1}  \chi(\xi)\frac{1}{\tilde \chi(\xi)\tilde \Psi(\xi)+\tilde{\tilde \chi}(\xi)} \in \mathcal A.
\end{align*}
We can do similar argument around $\xi=\pi$ if $W_\pi=0$.
Therefore, we have $t\in \mathcal A$.
Similarly, we have $r_\pm \in \mathcal A$.
\end{proof}

By the regularity of Jost solutions, we can show the finiteness of eigenvalues.
\begin{proof}[Proof of Proposition \ref{prop:finite}]
Since discrete spectrum can only accumulate at the edge of the essential spectrum, it suffices to show this does not happen.
Thus, it suffices to show $W_{\im \delta}\neq 0$ and $W_{\pi+\im \delta}\neq 0$ for $0<\delta<\delta_0$ for sufficiently small $\delta_0$ (we also have to consider the edge of $\sigma_-$ but it will be the same).

To show $W_{\im \delta}\neq 0$ for the generic case, it is enough to show $W_\cdot$ is continuous in $\T_{\delta_0,\geq 0}$ for some $\delta_0$ because $W_0W_\pi\neq 0$ in this case.
However, the continuity follows from the definition of $W_\xi$, Proposition \ref{prop:19} and the fact that $\mathcal A\hookrightarrow L^\infty(\T)$.

For the exceptional case, without loss of generality, it suffices to show $W_\xi$ is not zero near $\xi=0$.
By \eqref{54} and \eqref{57}, we see $W_\xi \neq 0$ near $0$.
\end{proof}

\section{Proof of Theorem \ref{thm:main}}\label{sec:6}

Before the proof of Theorem \ref{thm:main}, we introduce the van der Corput lemma by Egorova, Kopylova and Teschl \cite{EKT15JST}.

\begin{lemma}[Lemma 5.1 of \cite{EKT15JST}]\label{lem:21}
Let $-\pi\leq a< b\leq\pi$.
Set $I(t):=\int_a^b e^{\im t \phi(\xi)}g(\xi)\,d\xi$.
Assume $|\phi^{(m)}(\xi)|\gtrsim 1$ for some $m\geq 2$ where $\phi^{(m)}$ is the $m$th derivative of $\phi$.
Then, we have
\begin{align*}
|I(t)|\lesssim t^{-1/m}\|g\|_{\mathcal A}.
\end{align*}
\end{lemma}

We are now in the position to prove the main theorem of this paper.

\begin{proof}[Proof of Theorem \ref{thm:main}]
First, as \eqref{25.2} by Lemma \ref{lem:13.1}, we have
\begin{align*}
&J_{\mathrm{VE}}U^t P_+(U)J_{\mathrm{EV}}(x,y)=-\frac{1}{2\pi}\int_\T e^{\im t \(\lambda(\xi)+\frac{|x-y|}{t}\xi\)} \frac{\lambda'(\xi)}{W_{\xi}}\\&\quad 
\times\(m_-(x,\xi){}^tm_+(y,\xi)\sigma_1\begin{pmatrix} 1_{\leq y}(x) & 0 \\ 0 & 1_{<y}(x)\end{pmatrix}+m_+(x,\xi){}^tm_-(y,\xi) \sigma_1\begin{pmatrix} 1_{> y}(x) & 0 \\ 0 & 1_{\geq y}(x)\end{pmatrix}\).
\end{align*}
Without loss of generality, we can assume $x\geq y$. Then we have 
\begin{align*}
&J_{\mathrm{VE}}U^t P_+(U)J_{\mathrm{EV}}(x,y)\\&=-\frac{1}{2\pi}
\int_\T e^{\im t \(\lambda(\xi)+\frac{|x-y|}{t}\xi\)} \frac{\lambda'(\xi)}{W_{0,\xi}}t(\xi)\(m_-(x,\xi){}^tm_+(x,\xi)\sigma_1\begin{pmatrix}\delta_{x,y} & 0\\ 0 & 0 \end{pmatrix}	+m_+(x,\xi){}^tm_-(y,\xi)\sigma_1\) \,d\xi,
\end{align*}
where, we have substituted the relation $W_\xi^{-1}=t(\xi)W_{0,\xi}^{-1}$ given in \eqref{42.2}.

For the case $y\leq 0\leq x$, by Proposition \ref{prop:19}, $$m_-(x,\xi){}^tm_+(x,\xi)\sigma_1\begin{pmatrix}\delta_{x,y} & 0\\ 0 & 0 \end{pmatrix}	+m_+(x,\xi){}^tm_-(y,\xi)\sigma_1$$ is uniformly bounded in $\mathcal A$
(Notice that for the first term is $0$ expect $x=y=0$).
Further, for $\phi(\xi)=\lambda(\xi)+\frac{|x-y|}{t}\xi$, we have $\min(|\phi''(\xi)|,|\phi'''(\xi)|)\gtrsim 1$ by \eqref{25.3} (see \cite{MSSSS18DCDS}), and $\lambda'(\xi)W_{0,\xi}^{-1}\in C^\omega(\T;\R)\subset \mathcal A$ by \eqref{25.1} and \eqref{25.3}, $t\in \mathcal A$ by Proposition \ref{prop:20}.
Therefore, we have the estimate from Lemma \ref{lem:21} in this case.

For the case $0\leq y\leq x$, $m_-(y,\xi)$ may not be bounded.
However, by \eqref{40}, we have
\begin{align*}
t(\xi) m_-(y,\xi)=m_+(y,-\xi)+r_+(\xi)e^{2\im \xi x}m_+(y,\xi).
\end{align*}
Thus, we have the estimate from Lemma \ref{lem:21} combined with Propositions \ref{prop:19} and  \ref{prop:20}.
The case $y\leq x\leq 0$ can also be handle by the same argument.
\end{proof}

\appendix

\section{Proof of Lemma \ref{lem:4.2}}\label{proof:tri}

We prove Lemma \ref{lem:4.2}. 
\begin{proof}[Proof of Lemma \ref{lem:4.2}]
We first show $\cos:\T_+\to \C\setminus[-1,1]$ is a biholomorphism.
We track the orbit $C_\delta=\{\cos (s+\im \delta)\ |\ s\in\R\}$ for $\delta>0$.
Since
\begin{align}
\cos (s+\im \delta)&=\cosh \delta \cos s  - \im \sinh \delta \sin s,\label{a22}
\end{align}
we see that this orbit is an ellipse which encircles around $[-1,1]$ in clockwise direction (when the pameter $s$ is increasing).
It is easy to see that $\cos$ is a surjection from $\Tp $ to $\C\setminus[-1,1]$.
To show that it is an injection, we first claim that for $0<\delta_1<\delta_2$, two orbits $C_{\delta_1}$, $C_{\delta_2}$ do not intersect.
Suppose $\cos(s_1+\im\delta_1)=\cos(s_2+\im\delta_2)$ for $0<\delta_1<\delta_2$.
Then,
\begin{align*}
\cos s_1=\frac{\cosh \delta_2}{\cosh \delta_1}\cos s_2,\quad \sin s_1=\frac{\sinh \delta_2}{\sinh \delta_1}\sin s_2.
\end{align*}
Thus,
\begin{align*}
1=\cos^2 s_1+\sin^2 s_1=1+\(\frac{\cosh^2 \delta_2}{\cosh^2 \delta_1}-1\)\cos^2 s_2+\(\frac{\sinh^2 \delta_2}{\sinh^2 \delta_1}-1\)\sin^2 s_2.
\end{align*}
However, r.h.s.\ is always larger than $1$, which is absurd.
Finally, $\partial_s\cos(s+\im \delta)\neq 0$ for all $s$, so we have the claim, since $\cos$ is holomorphic, it becomes automatically biholomorphic if it is an injection.

We next consider the extension of $\cos^{-1}$ through the cut $(-1,1)$. 
Notice that $\cos$ maps $\T_{+,\delta}:=\{z\in \Tp\ |\ 0<\Im z<\delta\}$ into $D_\delta$, where $D_\delta$ is the open set with the boundary consisting of $[-1,1]$ and $C_\delta$.
Therefore, if $\xi_n\to \xi\in (-1,1)$, for arbitrary $\delta>0$ we have $\xi_n\in D_\delta$ for sufficiently large $n$ and therefore $\cos^{-1}\xi_n \in \T_{+,\delta}$.
This implies $\Im \cos^{-1}\xi_n\to 0$ as $n\to \infty$.
Now, assume $\xi_n\to \xi\in (-1,1)$ and $\Im \xi_n>0$.
Set $\cos^{-1}\xi_n=s_n+\im \delta_n$.
Then, we have $\delta_n\to 0$.
Since $\Im \xi_n>0$, by \eqref{a22}, we see $-\pi<s_n<0$.
By \eqref{a22}, we have
\begin{align*}
|\Re \xi_n-\xi|=|\cosh \delta_n \cos s_n-\xi|.
\end{align*}
Thus, $|\cos s_n-\xi|\to 0$ as $n\to \infty$ and we have $s_n\to -\mathrm{arccos}\xi$.
Therefore we see that $\cos^{-1}$ can be continuously extended to $(-1,1)$ from above and it becomes $\R$ valued in $(-1,1)$.
By a similar manner we can also continuously extend $\cos^{-1}$ from below. 
Thus, we can extend $\cos^{-1}$ though the cut $[-1,1]$ from above by setting 
\begin{align*}
\cos^{-1}_+ z:=\begin{cases} \cos^{-1} z & \Re z\in (-1,1),\ \Im z\geq 0,\\
\overline{\cos^{-1}\bar z}&\Re z\in (-1,1),\ \Im z< 0,
\end{cases}
\end{align*}
and from below by setting
\begin{align*}
\cos^{-1}_- z:=\begin{cases} \overline{\cos^{-1}\bar z}  & \Re z\in [-1,1],\ \Im z> 0,\\
   \cos^{-1} z&\Re z\in [-1,1],\ \Im z\leq 0.
\end{cases}
\end{align*}
By the Schwarz reflection principle (see, Theorem 5.6 of \cite{SSComplex}), $\cos^{-1}_\pm$ is holomorphic in
$\{ z\in \C |\ |\Im z|<\delta_0\}$ and further $\cos\(\cos^{-1}_\pm z\)=z$.
\end{proof}

\section{Decoupling when $\beta=0$ and gauge transformation}\label{sec:decoup}
When $\beta=0$ ($\Leftrightarrow |{\alpha}|=1$), the system decouples.
Set $\chi_{x_0,+}\in\mathcal L(\mathcal H)$ by
\begin{align*}
\(\chi_{x_0,+}u\)(x)=\begin{cases} u(x) & x>x_0\\ \begin{pmatrix}0 \\  u_{\downarrow}(x_0)\end{pmatrix} & x=x_0\\ 0 & x<x_0 \end{cases}\quad \text{where}\quad u(x)=\begin{pmatrix} u_{\uparrow}(x) \\  u_{\downarrow}(x)\end{pmatrix}.
\end{align*}
Obviously, $\chi_{x_0,+}$ is an orthogonal projection on $\mathcal H$.

\begin{proposition}\label{prop:decouplingbyref}
	Let $\beta(x_0)=0$. Then, $[U,\chi_{x_0,+}]=0$.
\end{proposition}

\begin{proof}
	Take any $u\in \mathcal H$.
	It suffices to show that $\([U,\chi_{x_0,+}]u\)(x)=0$ for all $x\in \Z$.
	We have
	\begin{align}\label{eq:decoup1}
	\([U,\chi_{x_0,+}]u\)(x)=\begin{pmatrix} (C\chi_{x_0,+}u)_{\uparrow}(x-1)\\ (C\chi_{x_0,+}u)_{\downarrow}(x+1)
	\end{pmatrix}-\(\chi_{x_0,+} \begin{pmatrix} (Cu)_{\uparrow}(\cdot-1)\\ (Cu)_{\downarrow}(\cdot+1)
	\end{pmatrix} \)(x),
	\end{align}
	and
	\begin{align}\label{eq:decoup2}
	(C\chi_{x_0,+}u)_{\uparrow}(x)=\begin{cases}(Cu)_{\uparrow}(x) & x\geq x_0\\ 0 & x<x_0 \end{cases},\quad (C\chi_{x_0,+}u)_{\downarrow}(x)=\begin{cases}(Cu)_{\downarrow}(x) & x> x_0\\ 0 & x \leq x_0 \end{cases},
	\end{align}
	where we have used $\beta(x_0)=0$.
	
	First, if $x+1\leq x_0$, both term of the r.h.s.\ of \eqref{eq:decoup1} are $0$.
	Next, if $x-1\geq x_0$, then the two terms are the same, so they cancel each other out.
	If $x=x_0$, by \eqref{eq:decoup2} we have
	\begin{align*}
	\([U,\chi_{x_0,+}]u\)(x_0)&=\begin{pmatrix} (C\chi_{x_0,+}u)_{\uparrow}(x_0-1)\\ (C\chi_{x_0,+}u)_{\downarrow}(x_0+1)
	\end{pmatrix}-\(\begin{pmatrix} 0 \\ (Cu)_{\downarrow}(x_0+1) \end{pmatrix}\)=0.
	\end{align*}
	Therefore, we have the conclusion.
\end{proof}

By Proposition \ref{prop:decouplingbyref}, we see that the component of $\mathrm{supp}\chi_{x_0,+}$ never interact with the component of $\mathrm{supp}\chi_{x_0,-}$ where $\chi_{x_0,-}=1-\chi_{x_0,+}$.
That is,
\begin{align*}
\chi_{x_0,\mp}U^t \chi_{x_0,\pm}u_0=0.
\end{align*}
Therefore, one can consider the dynamics of quantum walks separately in $\mathrm{supp}\chi_{x_0,\pm}$.

\section*{Acknowledgments}  
M.M.\ was supported by the JSPS KAKENHI Grant Numbers JP15K17568, JP17H02851 and\\ JP17H02853, 19K03579, G19KK0066A,.
H.S. was supported by JSPS KAKENHI Grant Number JP17K05311.
E.S. acknowledges financial support from JSPS 
the Grant-in-Aid of Scientific Research (C) Grant No. 19K03616 and Reserch Origin for Dressed Photon.
A. S. was supported by JSPS KAKENHI Grant Number JP26800054\\ and JP18K03327. 
K.S acknowledges JSPS the Grant-in-Aid for Scientific Research (C) 26400156 and 18K03354.

\medskip

Masaya Maeda, Hironobu Sasaki

Graduate School of Science,
Chiba University,
Chiba 263-8522, Japan

{\it E-mail Address}: {\tt maeda@math.s.chiba-u.ac.jp, sasaki@math.s.chiba-u.ac.jp}

\medskip

Etsuo Segawa

Graduate School of Information Sciences, 
Tohoku University,
Sendai 980-8579, Japan

{\it E-mail Address}: {\tt e-segawa@m.tohoku.ac.jp}

\medskip

Akito Suzuki

Division of Mathematics and Physics,
Faculty of Engineering,
Shinshu University,
Nagano 380-8553, Japan

{\it E-mail Address}: {\tt akito@shinshu-u.ac.jp}

\medskip

Kanako Suzuki

College of Science, Ibaraki University,
2-1-1 Bunkyo, Mito 310-8512, Japan

{\it E-mail Address}: {\tt kanako.suzuki.sci2@vc.ibaraki.ac.jp}

\end{document}